\documentclass[sigconf,natbib=false]{acmart}
\usepackage[backend=biber]{biblatex}
\newif\ifcommentson

\commentsonfalse
\newcommand{\marienote}[1]{\ifcommentson\textcolor{purple}{Marie: #1}\fi}
\newcommand{\weiqinote}[1]{\ifcommentson\textcolor{blue}{WeiqiRes: #1}\fi}
\newcommand{\lc}[1]{\ifcommentson\textcolor{red}{\textbf{LC:} #1}\fi}

\newcommand{\toolname}{ESBMC}
\newcommand{\baseline}{CBMC}

\usepackage{graphicx}
\usepackage{tikz}
\usetikzlibrary{arrows.meta,positioning,calc,fit}
\graphicspath{{../figures/}{figures/}}

\usepackage{listings}
\usepackage{multirow}

\makeatletter
\define@key{dialogbox}{title}{\def\dialogbox@title{#1}}
\newsavebox{\dialogbox@box}
\newenvironment{dialogbox}[1][]{%
	\def\dialogbox@title{}%
	\setkeys{dialogbox}{#1}%
	\begin{lrbox}{\dialogbox@box}%
		\begin{minipage}{0.93\linewidth}%
			\vskip 6pt%
			\ifx\dialogbox@title\empty\else
				\colorbox{black!85}{%
					\makebox[\dimexpr\linewidth-2\fboxsep][l]{%
						\textcolor{white}{\bfseries\hspace{4pt}\dialogbox@title}%
					}%
				}%
				\vskip 1pt\leavevmode\\%
			\fi
			\ignorespaces
			}{%
			\vskip 6pt%
		\end{minipage}%
	\end{lrbox}%
	\par\vskip\medskipamount\noindent\shadowbox{\usebox{\dialogbox@box}}\par\vskip\medskipamount
}
\makeatother
\usepackage{fancyvrb}

\usepackage{fancybox}
\usepackage{keyval}
\usepackage{algorithm}



\usepackage{amsmath}
 \usepackage{amssymb}
\usepackage{amsthm}
\usepackage{booktabs}
\usepackage{array}
\usepackage{caption}
\usepackage{hyphenat}

\newtheorem{definition}{Definition}
\newtheorem{proposition}{Proposition}
\newtheorem{theorem}{Theorem}
\newtheorem{remark}{Remark}

\AtBeginDocument{%
  }

\acmConference[Preprint]{Preprint}{October 2026}{arXiv}
\renewcommand\footnotetextcopyrightpermission[1]{}
\acmYear{2026}
\copyrightyear{2026}

\author{Weiqi Wang}
\orcid{0009-0005-8247-5994}
\affiliation{%
  \institution{University of Manchester}
  \city{Manchester}
  \country{UK}}
\email{Weiqi.Wang-2@postgrad.manchester.ac.uk}

\author{Marie Farrell}
\orcid{0000-0001-7708-3877}
\affiliation{%
  \institution{University of Manchester}
  \city{Manchester}
  \country{UK}}
\email{Marie.Farrell@manchester.ac.uk}

\author{Liping Zhao}
\orcid{0000-0001-8556-8655}
\affiliation{%
  \institution{University of Manchester}
  \city{Manchester}
  \country{UK}}
\email{Liping.Zhao@manchester.ac.uk}

\author{Lucas C. Cordeiro}
\orcid{0000-0002-6235-4272}
\affiliation{%
  \institution{University of Manchester}
  \city{Manchester}
  \country{UK}}
\email{Lucas.Cordeiro@manchester.ac.uk}

\title{Function Contracts for Bounded Model Checking of C Programs}

\ccsdesc[500]{Theory of computation~Verification by model checking}
\ccsdesc[300]{Theory of computation~Hoare logic}
\ccsdesc[300]{Theory of computation~Automated reasoning}
\ccsdesc[100]{Theory of computation~Separation logic}

\keywords{Bounded Model Checking, Function Contracts, Modular Verification,
  Frame Rule}

\begin{document}
\begin{abstract}
Bounded model checking unrolls a program into a single formula, so a function
holding a loop costs the solver a fresh copy of that loop at every call site.
Function contracts break that coupling. Verify the body once against its
contract, then replace each call by a stub derived from it. \baseline{} supports
this through dynamic frame condition checking, applied at industrial
scale~\cite{cbmc-s2n}, but with no formal account of why the replacement is
sound and no characterisation of what the frame check costs.

We present the function-contract framework for \toolname{}. The clause syntax,
both phases, and the first \texttt{assigns} enforcement the tool has had are all
introduced here. Its technical core is an observational frame rule: snapshot the
locations the contract does not permit the function to write, then compare them
at exit. The encoding emits a number of verification conditions set by how many
globals the program has, not by how many writes the body performs, while dynamic
frame condition checking grows with the write count. That is a property of the
encoding rather than a timing, so it moves with neither the machine nor the
solver. The rule itself is the classical one; what we add is its bounded form
inside a BMC encoding, with the pipeline's algebraic skeleton machine-checked in
Lean~4 and the limits of that model stated.

We evaluate on synthetic sweeps, SV-COMP tasks, an ANSSI X.509 parser, the
FreeRTOS \texttt{heap\_1} allocator, and the ML-KEM reference implementation
whose \baseline{} contracts we reuse unchanged. At the coefficient count
FIPS-203 mandates, neither tool finishes on the monolithic route inside ten
minutes, and both finish on the modular one. The framework is merged into
\toolname{} \texttt{master}. Contracts are what is left when unrolling runs out
of room.
\end{abstract}

\maketitle

\section{Introduction}
\label{sec:introduction}
\marienote{similar to the abstract, the introduction is not very accessible from the beginning. You need to ease into it a bit more to draw the reader in. In general, the paper feels a bit more like a technical report than a research paper so be careful of this and focus on the main points that you want the reader to take away.}
Bounded model checking (BMC)~\cite{biere1999symbolic} of C programs suffers from two compounding scalability bottlenecks. The first is loop unrolling: a BMC tool must materialise up to $N$ unrolled copies of each loop body into a single logical formula, so the formula grows linearly in the loop bound~$N$. The second is call-site replication: when a non-trivial function is called $K$~times, the solver reasons about $K$ independent copies of the already unrolled body, giving $O(N \times K)$ growth.

The effect is not confined to hard functions. \toolname{}~\cite{gadelha2018esbmc}, a mature open-source BMC tool for C, exceeds a 120-second timeout on our synthetic benchmark at $N{=}20$, $K{=}10$. At the real ML-KEM size ($N{=}256$) it exceeds a 600-second one. \baseline{}~\cite{kroening2014cbmc}, an industrial-strength BMC tool applied at scale at Amazon Web Services~\cite{cbmc-s2n}, needs $599.7$\,s for the same task, so neither tool has room left on the monolithic route. The functions in question are not intrinsically complex.
\lc{This bounds formula size; measured solver time is superlinear, so separate the two claims.} \weiqinote{Addressed: this sentence now bounds \emph{formula size} only. The superlinear growth of measured solver time is stated separately in RQ2 (\S\ref{sec:rq2}) and is the reason $K^*$ is presented as a lower bound rather than an exact prediction.}
\lc{Name the exact benchmark and configuration here; ``our benchmark'' is too vague for a headline claim.} \weiqinote{Addressed: the exact configuration is now named ($N{=}20$, $K{=}10$ on the synthetic \texttt{compute} benchmark, 120\,s budget) together with the real ML-KEM size ($N{=}256$). A factual error was also fixed here: \baseline{} is applied at scale at AWS, not developed there.}


Function contracts are the established remedy~\cite{cbmc-s2n, kroening2014cbmc}. A contract comprises a precondition (\texttt{requires}), a postcondition (\texttt{ensures}), and an assigns clause naming the locations that the function may modify. Verification splits into two independent phases. \emph{Enforce}: verify that the function implementation satisfies its contract in isolation with nondeterministic inputs, an $O(N)$ problem that is independent of the number of
call sites. \emph{Replace}: at each of the $K$ call sites, substitute the function body with a nondeterministic stub that havoces the declared assigns targets and assumes the postcondition, with no loop
unrolling. The total cost $O(N) + K \cdot c_Q$, where $c_Q$ is the per-site cost of discharging the postcondition \marienote{say what this means in practice, what is the `cost' in terms of?}, breaks the $O(N \times K)$ scaling once the one-time Enforce cost is amortised \marienote{I am not sure what `amortised' means} over enough call sites: at $N=10$, $K=10$, modular verification completes in 1.8\,s against 42.0\,s for direct, a
$23{\times}$ speedup. At larger $N{\times}K$ direct \marienote{direct? as opposed to indirect?} verification times out while
modular still terminates. \marienote{the last sentence is confusing, try setting the context a bit better. This relates to my overall point about accessibility.}

The assigns clause is the part that must be enforced, and enforcing it means establishing that the function modifies only its declared locations. Two semantics are available. Operational semantics monitors every write during execution, as DFCC does~\cite{kroening2014cbmc, cbmc-contracts}. Observational semantics compares the final state to the initial one, which is the standard relational reading in the Java Modeling Language (JML)~\cite{burdy2005overview} and the ANSI/ISO C Specification Language (ACSL)~\cite{baudin2013acsl}. The choice determines how many verification conditions (VCCs) the encoding emits, and therefore what the solver has to discharge.
\lc{Fix \texttt{baudin2013acsl} authors; ``Hubert, Thierry'' is an acknowledgee, replace with ``Cuoq, Pascal''.} \weiqinote{Addressed: \texttt{baudin2013acsl} now lists ``Cuoq, Pascal'' as author; ``Hubert, Thierry'' removed.}

\marienote{the below two paragraphs are nice and clear :-) }
\baseline{} supports contracts through Dynamic Frame Condition Checking (DFCC)~\cite{cbmc-contracts}, applied at industrial scale~\cite{cbmc-s2n}. Prior to
this work, \toolname{} had no contract support of any kind: no clause syntax, no
\emph{Enforce}, no \emph{Replace}, and no \texttt{assigns} enforcement. The clause
syntax and both phases are unpublished work of ours, and this paper presents them
with the frame rule they were missing. Here, we build and formalise the complete contract framework for \toolname{}, using \baseline{}/DFCC as the comparison baseline. \baseline{}/DFCC is a mature industrial implementation, but it comes with no formal account. Soundness therefore rests on engineering intuition rather than proof, the VCC cost of observational against operational assigns has never been characterised, and a practitioner choosing between the two encodings has no criterion to go on. 
\lc{Retarget the CBMC contracts citation to the assigns clause page carrying the per write DFCC text.} \weiqinote{Addressed: \texttt{cbmc-contracts} now points at the assigns-clause page carrying the per-write DFCC description.}

This paper presents the complete function contract framework for \toolname{}. Its technical core is the Observational Frame Rule. Under the observational assigns semantics of JML and ACSL, the \texttt{frame\_enforcer} module emits $(G{-}M)+E$ VCCs per Enforce step, for $G$ tracked globals, $M$ assigns targets and $E$ postconditions, independent of the number of write instructions~$W$ in the body, and the same module serves both Enforce and call-site abstraction.
The frame rule and the Enforce/Replace decomposition are classical
(\S\ref{sec:related_work}); we claim neither. The contributions are the
following.

\marienote{each of these enumerated points should be shorter - try to condense the main point in each}
\begin{enumerate}
  \item \textbf{The contract stack for \toolname{}}~(\S\ref{sec:methodology}, \S\ref{sec:experiments}). The clause syntax, Enforce, Replace and \texttt{assigns} enforcement are implemented as one \texttt{frame\_enforcer} module and merged into \toolname{} \texttt{master}. We evaluate on synthetic benchmarks, SV-COMP tasks, and three real-world case studies: ANSSI X.509, FreeRTOS \texttt{heap\_1}, and the NIST FIPS-203 ML-KEM (Kyber) reference implementation, whose third-party \baseline{} contracts we reuse. \baseline{}/DFCC~\cite{kroening2014cbmc, cbmc-contracts} is the comparison baseline throughout.
  \item \textbf{Assigns encoding and complexity model}~(\S\ref{sec:assigns}, \S\ref{sec:model}). We show that the standard observational assigns semantics of JML/ACSL admits an efficient BMC encoding generating $(G{-}M)+E$ VCCs, independent of the write instruction count~$W$, versus DFCC's $20{+}W{+}E$. The crossover $G < 20{+}M{+}W$ then says when the observational encoding emits
fewer VCCs. We also give the breakeven call count~$K^*$ at which modular
verification amortises its one-time Enforce cost. \marienote{amortises? again, not a very intuitive word to use}
  \item \textbf{A bounded, machine-checked account}~(\S\ref{sec:pipeline}). We state the algebraic skeleton of the Enforce/Replace pipeline over an abstract state-transformer model and check it in Lean~4~\cite{moura2021lean}: five theorems and one proposition (Theorems~\ref{thm:enforce}--\ref{thm:hoare_frame}, Proposition~\ref{prop:ptr}) covering enforce-transformation correctness, call-site abstraction, observational assigns compliance, Hoare frame rule embedding, contract substitution, and the pointer extension. They are bounded realisations of the classical modular-soundness results inside a BMC encoding, and they compose into soundness relative to the unwinding bound, meaning no false negatives against monolithic verification at the same bound. That composition is Proposition~\ref{prop:relsound}, argued on paper from the Lean theorems rather than itself a Lean theorem. Remark~\ref{rem:lean-scope} states what the model leaves out: loops, the BMC encoding, and the C memory model. We corroborate the pipeline against the C++ implementation empirically via seeded-fault agreement (\S\ref{sec:ce}).
\end{enumerate}
\lc{Only \texttt{ptr} is Lean checked; \texttt{relsound}, \texttt{op$\Rightarrow$obs} are not, so fix ``two propositions in Lean''.} \weiqinote{Addressed: the claim is now exactly five theorems and one proposition mechanised in Lean~4. Proposition~\ref{prop:relsound} is explicitly described as the composed corollary \emph{argued on paper from the Lean theorems rather than itself a Lean theorem}, and Remark~\ref{rem:lean-scope} states what the model does not capture (loops, the BMC encoding, the C memory model).}
\lc{The constant 21 is CBMC 6.8.0 specific; state that the crossover is tool version dependent.} \weiqinote{Addressed: ``tool-agnostic'' is removed from the conclusion; the constant is now stated as specific to \baseline{}~6.8.0, with only the \emph{shape} (flat versus linear) claimed to transfer. The criterion is also re-scoped as a within-encoding VCC accounting rather than a cross-tool wall-clock predictor, since a VCC is not a comparable unit of solver work across the two backends.}
\marienote{make sure that you say everything is in a git repo somewhere and provide the link. It is good to point this out in the introduction, conclusion and maybe at least one other place in the paper.}

\section{Background}
\label{sec:background}

\subsection{Function Contracts: A Running Example}
\marienote{you have line numbers in the figure so use them when you are describing it in the text, this will make it easier for the reader to follow}
Function contracts are specified via built-in macros that are embedded in the function body. \baseline{} has long supported this through its DFCC pipeline~\cite{kroening2014cbmc, cbmc-contracts} which is inspired by Dynamic Frames~\cite{kassios2006dynamic}. \toolname{}
had no contract support before this line of work; the clause syntax,
both phases and the \texttt{assigns} enforcement are all introduced here. This paper contributes the complete framework. Listing~\ref{lst:compute} illustrates our running example: \texttt{compute()} clamps each element of an integer array to a range $[\mathit{lo}, \mathit{hi}]$, returns the accumulated sum, and records the processed count in a global \texttt{g\_processed} that it declares in its \texttt{assigns} clause, so the example exercises frame
checking as well as pre- and postconditions.

\begin{figure}[t]
\begin{lstlisting}[language=C]
int g_processed; /* global the contract lets compute() modify */

int compute(int *arr, int n, int lo, int hi) {
  __ESBMC_requires(
      __ESBMC_is_fresh(arr, n*sizeof(int))
      && n > 0 && lo >= 0 && lo <= hi
      && hi <= 100);
  __ESBMC_assigns(g_processed); /* frame: only g_processed */
  __ESBMC_ensures(
    __ESBMC_return_value >= n * lo &&
    __ESBMC_return_value <= n * hi);
  int sum = 0;
  for (int i = 0; i < n; i++) {
    int v = arr[i];
    if (v < lo) v = lo;
    if (v > hi) v = hi;
    sum += v;
  }
  g_processed = n; /* the only global write */
  return sum;
}
\end{lstlisting}
\caption{Running example: \texttt{compute()} with an \toolname{} contract. The \texttt{assigns(g\_processed)} clause exercises frame checking: the function
may modify \texttt{g\_processed} and no other global.}
\label{lst:compute}
\Description{A description of the image for accessibility.}
\end{figure}

Verification proceeds in two phases. \emph{Enforce}: \toolname{} generates a wrapper that assumes the precondition, calls the original body, then asserts the postcondition and assigns compliance. The loop is unrolled to the bound once, in isolation, instead of once per call
site. \emph{Replace}: at each call site, \toolname{} substitutes the body with a stub that havoces the declared assigns target (\texttt{g\_processed}) and assumes the postcondition; no loop unrolling occurs at the call site. The measured speedups for this example across loop bounds and call counts are reported in Section~\ref{sec:experiments} (Table~\ref{tab:rq2}). \marienote{in lots of places from here onwards you forward reference to other parts of the paper. In general this is a reasonable thing to do but you are doing it so much that the reader forgets what the current section they are reading is about because there are so many pointers to future parts. }

Table~\ref{tab:syntax} summarises the contract annotation syntax for both \toolname{} and \baseline{}, which we use as a reference comparison tool throughout the paper. \baseline{} requires a 3-step pipeline (\texttt{goto-cc} $\to$ \texttt{goto-instrument} $\to$ \texttt{cbmc}) for each enforce or replace step; \toolname{} uses a single invocation.

\begin{table}[t]
\small
\caption{Contract annotation syntax: \toolname{} vs.\ \baseline{}.}
\label{tab:syntax}
\begin{tabular}{lll}
\toprule
\textbf{Feature} & \textbf{\toolname{}} & \textbf{\baseline{}} \\
\midrule
Precondition  & \texttt{\_\_ESBMC\_requires(p)}  & \texttt{\_\_CPROVER\_requires(p)} \\
Postcondition & \texttt{\_\_ESBMC\_ensures(e)}   & \texttt{\_\_CPROVER\_ensures(e)} \\
Assigns       & \texttt{\_\_ESBMC\_assigns(t)}   & \texttt{\_\_CPROVER\_assigns(t)} \\
Return value  & \texttt{\_\_ESBMC\_return\_value} & \texttt{\_\_CPROVER\_return\_value} \\
Old value     & \texttt{\_\_ESBMC\_old(x)}       & \texttt{\_\_CPROVER\_old(x)} \\
\midrule
Enforce/Replace & Single \texttt{esbmc} call & 3-step pipeline \\
\bottomrule
\end{tabular}
\end{table}

\subsection{The Frame Problem in BMC}

In modular verification, when a function is replaced by its contract stub, the caller learns which memory locations may have changed. The assigns clause encodes this information: locations not listed are guaranteed to remain unchanged. Enforcing this guarantee is the \emph{frame problem} in program verification~\cite{reynolds2002separation, o2019separation}.
There are two natural encodings for assigns clause enforcement in BMC:

\paragraph{(1) Observational encoding.} Snapshot each global $v_i \notin S$ at function entry ($s_i \gets v_i$), then assert equality at exit ($v_i = s_i$). Total VCC count: $(G - M) + E$, independent of the number of write instructions~$W$ in the function body.

\paragraph{(2) Operational encoding (DFCC)} Before every write instruction in the body, test membership in the declared write-set via a runtime data structure. Total VCC count: $20 + W + E$ (20 library infrastructure VCCs plus one per write instruction plus one per postcondition). This is the approach taken by DFCC in \baseline{}~\cite{kroening2014cbmc, cbmc-contracts}.
\lc{Formula gives 23 at $E{=}1$, yet Table~\ref{tab:rq1} reports 22; state whether the 21 already includes ensures.} \weiqinote{This was the most valuable comment in the set, and the answer is that the \emph{formula} was wrong, not the table. Re-measuring \baseline{}~6.8.0 over a $W \times E$ sweep gives, for $W \ge 1$, exactly $20{+}W{+}E$, not $21{+}W{+}E$: $(W{=}1,E{=}0) \to 21$, $(W{=}1,E{=}1) \to 22$, $(W{=}2,E{=}1) \to 23$. The $W{=}0$ empty body sits at~16, off the line; the paper had been mistaking the $(W{=}1,E{=}0)$ point for the intercept. So Table~\ref{tab:rq1}'s 22 was right all along. Propagated throughout: the crossover is $G < 20{+}M{+}W$ and the empirical crossover point is $G{=}22$, not~23; Figure~\ref{fig:crossover} regenerated.}

The two encodings answer different questions. Observational asks: \emph{did any non-assigned location change by the time the function returned?} Operational asks: \emph{was any non-assigned location written during execution?} These differ only for the \emph{temporary-write} pattern, in which a write and restore before exit, which is irrelevant to the caller in a sequential setting (Section~\ref{sec:scope}). In this work we adopt observational semantics \marienote{say what you mean by `observational semantics' and why you adopt it}, which is the standard assigns semantics of JML~\cite{burdy2005overview} and ACSL~\cite{baudin2013acsl}, and is sound for sequential single-threaded C programs.

\section{Framework for Function Contract Verification}
\label{sec:semantics}

\subsection{Contract Model and Verification Pipeline}
\label{sec:pipeline}
\lc{Add an overview figure here; input C plus contract, through Enforce then Replace, to verdict.} \weiqinote{Addressed: Figure~\ref{fig:overview} is the end-to-end overview: annotated C source and its $K$ call sites at the top, the two phases and what each emits in the middle, and the verdict at the bottom, with inputs, outputs, and the per-phase VCC cost labelled. It is drawn in TikZ rather than imported, so the operation labels stay editable.}

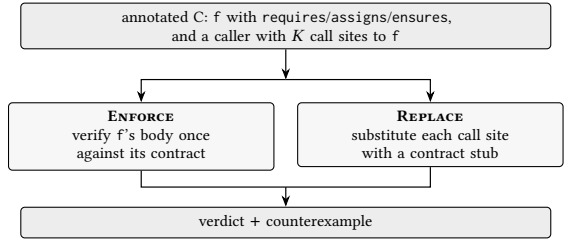
\begin{figure}[t]
\centering
\begin{tikzpicture}[
  font=\scriptsize,
  box/.style={draw,rounded corners=1.5pt,inner sep=3pt,line width=.4pt,align=center},
  io/.style={box,fill=black!7,text width=68mm},
  ph/.style={box,fill=black!3,text width=33mm},
  ar/.style={-{Stealth[length=1.5mm]},line width=.45pt},
]
\node[io] (src) {annotated C: \texttt{f} with \texttt{requires}/\texttt{assigns}/\texttt{ensures},\\
  and a caller with $K$ call sites to \texttt{f}};
\coordinate (a) at ($(src.south)+(0,-4mm)$);
\draw[ar] (src.south) -- (a);
\node[ph,anchor=north east] (e) at ($(a)+(-1.5mm,-3mm)$) {\textbf{\textsc{Enforce}}\\
  verify \texttt{f}'s body once\\ against its contract};
\node[ph,anchor=north west] (r) at ($(a)+(1.5mm,-3mm)$) {\textbf{\textsc{Replace}}\\
  substitute each call site\\ with a contract stub};
\draw[ar] (a) -| (e.north);
\draw[ar] (a) -| (r.north);
\coordinate (b) at ($(e.south)!0.5!(r.south)$);
\node[io,below=5mm of b] (v) {verdict $+$ counterexample};
\coordinate (m) at ($(v.north)+(0,2.6mm)$);
\draw[line width=.45pt] (e.south) |- (m);
\draw[line width=.45pt] (r.south) |- (m);
\draw[ar] (m) -- (v.north);
\end{tikzpicture}
\Description{End-to-end overview of the contract pipeline: annotated C source,
the Enforce and Replace phases, and the verdict.}
\caption{The contract pipeline end to end. Enforce is paid once; Replace is
applied at each of the $K$ call sites, where it removes the loop unrolling that
makes monolithic verification scale as $O(N{\times}K)$.}
\label{fig:overview}
\end{figure}

\begin{remark}[What the Formal Model Does and Does Not Cover]
\label{rem:lean-scope}
Our Lean~4 development reasons at the level of \emph{state transformers}: the
model is $\mathit{State} = (\mathit{vars}: \mathsf{String} \to \mathbb{Z},\,
\mathit{is\_err}: \mathbb{B})$ and an implementation $\mathit{impl}$ is an
\emph{abstract} total function $\mathit{State} \to \mathit{State}$. The theorems
below are therefore parametric in $\mathit{impl}$: they establish the algebraic
correctness of the Enforce/Replace \emph{transformation}, independently of how
$\mathit{impl}$ is realised. Two things are deliberately outside this
model. First, the model
does not formalise loop unrolling: BMC discharges the enforce check only up to an
unwinding bound (Section~\ref{sec:bounded}), so the end-to-end guarantee is
bounded, not unbounded as in deductive verification. Second, the Lean proofs
validate the \emph{algorithm}, not the \toolname{} C\texttt{++} code that
implements it; we do not claim the implementation is verified. Nothing in the
development refers to that code, and the code is not extracted from it. What
connects the two is the seeded-fault agreement of Section~\ref{sec:ce}, which
is a test rather than a proof. The scalar model
covers integer-typed globals; the pointer-aliasing extension
(Proposition~\ref{prop:ptr}) is machine-checked separately via an abstract
\texttt{Address} type.
\end{remark}

We model programs as state transformers over this memory model. A \emph{contract}
is a triple $(P, Q, S)$, where $P : \mathit{State} \to \mathbb{B}$ is the
precondition, $Q : \mathit{State} \to \mathit{State} \to \mathbb{B}$ is the
postcondition (relating pre- and post-states), and $S \subseteq \mathit{Vars}$ is
the assigns set, the variables that the function may modify.

\marienote{various formulas and text is running into the margin, be sure to fix this when you can}
The \emph{Enforce/Replace pipeline} decomposes modular verification into two phases:
\begin{itemize}
  \item[1.] \textbf{Enforce}: generate a wrapper $\mathit{wrap}(\mathit{impl}, P, Q, s_0)$ \marienote{this function needs a 1-sentence explanation/justification because it kind of seems to come out of nowhere} that assumes $P$, calls the real implementation, then asserts $Q$ and assigns compliance.
  \item[2.] \textbf{Replace}: at each call site, substitute $\mathit{impl}$ with a stub $\mathit{abs}(P, Q, s_0)$, defined as $\{\,s \mid P(s_0) \wedge Q(s_0, s)\,\}$, that havoces only the declared assigns targets and assumes $Q$. The implementation body is never examined.
\end{itemize}

\begin{theorem}[Enforce Transformation Soundness]
\label{thm:enforce}
Let $\mathit{wrap}(\mathit{impl},\allowbreak P,\allowbreak Q,\allowbreak s_0)$ run $\mathit{impl}(s_0)$ under
precondition $P$ and flag an error if $Q(s_0, \cdot)$ fails. For any $s_0$, with
$s_f = \mathit{wrap}(\mathit{impl},\allowbreak P,\allowbreak Q,\allowbreak s_0)$:
\[ P(s_0) \;\wedge\; \neg s_f.\mathit{is\_err} \;\Longrightarrow\; Q(s_0,\, s_f) \]
That is, whenever the wrapper raises no error, the postcondition holds. The wrapper is a
faithful checker of contract compliance, parametric in $\mathit{impl}$. \emph{Lean~4}: \texttt{phase1\_enforcement\_sound}. \marienote{explain the notation that you are using, also make sure you clarify what `\emph{Lean~4}: \texttt{phase1\_enforcement\_sound}' refers to. My understanding is that this is the label in the corresponding Lean proof but you need to make this clear for a reader who is not familiar with Lean.}
\end{theorem}

\noindent
We stress what Theorem~\ref{thm:enforce} does not say. Because
$\mathit{impl}$ is an abstract function, the theorem certifies the \emph{wrapper
construction}; it does not by itself establish that a finite-bound BMC run
discharges $\neg s_f.\mathit{is\_err}$ for every input when $\mathit{impl}$
contains loops. That bridge, from the algebraic statement to what a bounded run actually proves, is made precise in Section~\ref{sec:bounded}.

\begin{theorem}[Replace Soundness]
\label{thm:replace}
For any property $\mathcal{P}$:
\[ \bigl(\forall s.\; Q(s_0, s) \Rightarrow \mathcal{P}(s)\bigr) \;\Longrightarrow\; \forall s \in \mathit{abs}(P, Q, s_0).\; \mathcal{P}(s) \]
The stub ranges over every state satisfying $Q$, a superset of the
states the real implementation can reach, so any property provable against the
stub holds for the real function. This is what lets Replace-mode BMC reason about call
sites without re-examining the body. \emph{Lean~4}: \texttt{phase2\_abstraction\_sound}.
\end{theorem}

\begin{theorem}[Contract Substitution]
\label{thm:substitution}
If $\mathit{impl}$ satisfies its contract at $s_0$ and a property $\mathcal{P}$
holds for every state reachable via the stub, then $\mathcal{P}(\mathit{impl}(s_0))$:
\begin{align*}
  & P(s_0)  \wedge\; Q(s_0, \mathit{impl}(s_0)) \\
  & \wedge\; \bigl(\forall s' \in \mathit{abs}(P,Q,s_0).\;\mathcal{P}(s')\bigr) \;\Longrightarrow\; \mathcal{P}(\mathit{impl}(s_0))
\end{align*}
\emph{Lean~4}: \texttt{contract\_substitution\_sound}.
\end{theorem}

\subsection{Bounded Soundness of the Pipeline}
\label{sec:bounded}

Theorems~\ref{thm:enforce}--\ref{thm:substitution} are statements about an
abstract $\mathit{impl}$. We now state precisely what they buy when $\mathit{impl}$
is verified by \emph{bounded} model checking with unwinding bound~$b$.

A BMC enforce run does not establish $\neg s_f.\mathit{is\_err}$ for all inputs
under arbitrarily many loop iterations; it establishes it for all executions that
unroll within~$b$. \toolname{} emits an \emph{unwinding assertion} at each loop:
when these pass, no execution exceeds~$b$ and the enforce guarantee is
unconditional for that program; this is exactly the case for the
statically-bounded loops that dominate verification-oriented C (e.g.\ the
fixed $256$-coefficient loops of the ML-KEM case study in
Section~\ref{sec:mlkem}, or arrays bounded for BMC as in
Section~\ref{sec:svcomp}). When an unwinding assertion cannot be discharged, the
guarantee is explicitly \emph{up to depth~$b$}, as for BMC in general.

The contribution of the Lean theorems is then a \emph{relative-soundness} \marienote{need to explain what you mean by relative soundness, relative to what?}
guarantee, which is the property a modular verifier needs:

\begin{proposition}[No false negatives relative to the bound]
\label{prop:relsound}
Fix an unwinding bound~$b$ and suppose the Enforce phase reports success at~$b$.
Then for any caller property $\mathcal{P}$ over the non-assigned frame, every
counterexample that monolithic (inlined) BMC finds at~$b$ is also found by
Replace-mode BMC at~$b$: modularisation introduces no false negatives. Because the
stub over-approximates the implementation (it ranges over all states
satisfying~$Q$, not only the reachable ones), the converse need not hold: Replace
may additionally report a violation that the inlined program does not exhibit when
the contract is weaker than the implementation. Modular verification is thus sound
(no missed bugs) but not complete (a too-weak contract can raise a false alarm).
\end{proposition}

\noindent
This follows by composing Theorems~\ref{thm:enforce} and~\ref{thm:substitution}:
Enforce establishes (at~$b$) that $\mathit{impl}$ is a state satisfying $Q$;
Replace verifies $\mathcal{P}$ against all states satisfying $Q$, a
superset. Hence any bug direct BMC would find at bound~$b$ is preserved by the
modular pipeline at the same bound; the speedups in
Section~\ref{sec:experiments} reduce solver work without weakening the
guarantee (Enforce and Direct use the same~$b$ throughout). The one-directional
nature of this guarantee, and the false alarm a weak contract can raise, are both
corroborated empirically in Section~\ref{sec:ce} on the actual implementation.

To make the formalisation auditable from the paper itself,
Figure~\ref{fig:lean_formal} reproduces the core verbatim: the state model, the
Enforce wrapper with its soundness theorem, and the substitution theorem that
composes into Proposition~\ref{prop:relsound}. The development is 237 lines of
Lean~4 and contains \textbf{no \texttt{sorry}, no \texttt{admit}, and no added
\texttt{axiom}} (only the standard classical-logic instance); ``zero errors''
here means zero proof obligations left open, not merely a clean compile.
\lc{State the nontrivial formal content; over abstract \texttt{impl} these theorems risk reading as restatements.} \weiqinote{Partially addressed. Remark~\ref{rem:lean-scope} now states what the abstract model does and does not carry, and \S\ref{sec:ce} frames the seeded-fault study as the empirical bridge from the Lean model to the C++ implementation. The model itself is not strengthened. I think this is the weakest remaining part of the paper and would value your view on whether the scoped version is defensible.}

\begin{figure}[t]
\begin{lstlisting}[language=ML,basicstyle=\ttfamily\scriptsize]
structure State where           -- the whole memory model
  vars   : String -> Int        -- global integer store
  is_err : Bool                 -- error flag

-- Enforce: assume req, run impl, flag error if ens fails
def contract_wrapper (impl req ens) (s0) : State :=
  if req s0 then (let s1 := impl s0;
       if ens s0 s1 then s1 else {s1 with is_err := true})
  else s0

theorem phase1_enforcement_sound :        -- Enforce soundness
  (req s0 /\ not (contract_wrapper impl req ens s0).is_err)
    -> ens s0 (contract_wrapper impl req ens s0)

-- Replace stub: all states satisfying the postcondition
def abstract_transition (req ens) (s0) : State -> Prop :=
  fun s => req s0 /\ ens s0 s

theorem contract_substitution_sound :     -- composes to Prop 1
  req s0 -> ens s0 (impl s0)
    -> (forall s', abstract_transition req ens s0 s' -> P s')
    -> P (impl s0)
\end{lstlisting}
\caption{Verbatim core of the Lean~4 development (\texttt{ContractSoundness.lean};
no \texttt{sorry}/\texttt{admit}/\texttt{axiom}).
\texttt{phase1\_enforcement\_sound} is Enforce soundness;
\texttt{contract\_substitution\_sound} composes with it to yield the
relative-soundness guarantee of Proposition~\ref{prop:relsound}.}
\label{fig:lean_formal}
\end{figure}

\subsection{Assigns Clause Encoding}
\label{sec:assigns}

A contract's assigns set, $S$, constrains which memory locations the function may modify. Within the pipeline, this constraint must be formalised and checked. We present two natural encodings; our framework adopts the observational one. In the following, a \emph{program memory state} $\sigma$ is a total function from memory locations~$\ell$ (variables and heap cells) to values; $\sigma(\ell)$ denotes the value stored at location~$\ell$ in state~$\sigma$. $S$ is the declared assigns set.
\begin{definition}[Observational Assigns Semantics]
\label{def:observational}
A function $f$ satisfies $\mathsf{assigns}(S)$ \emph{observationally} iff:
\[ \forall \ell \notin S.\;\; \sigma_{\mathit{post}}(\ell) = \sigma_{\mathit{pre}}(\ell) \]
This is the standard semantics of assigns/modifies clauses in JML~\cite{burdy2005overview} and ACSL~\cite{baudin2013acsl}: only the \emph{final} state is compared to the \emph{initial} state.
\end{definition}



\begin{definition}[Operational Assigns Semantics]
\label{def:dfcc}
A function $f$ satisfies $\mathsf{assigns}(S)$ \emph{operationally} iff:
\[ \forall (\ell := v) \in \mathit{Writes}(f).\;\; \ell \in S \]
where $\mathit{Writes}(f)$ is the multiset of write instructions executed. This is the semantics that is enforced by Dynamic Frame Condition Checking (DFCC).
\end{definition}

\begin{proposition}
\label{prop:op-implies-obs}
Operational $\Rightarrow$ Observational, but not vice versa.
\end{proposition}
\begin{proof}
($\Rightarrow$) Suppose $f$ satisfies $\mathsf{assigns}(S)$ operationally. Then every $(\ell := v) \in \mathit{Writes}(f)$ has $\ell \in S$. For any $\ell \notin S$, no write instruction in the execution of $f$ modifies $\ell$; therefore the value of $\ell$ cannot change, giving $\sigma_{\mathit{post}}(\ell) = \sigma_{\mathit{pre}}(\ell)$. Hence $f$ satisfies $\mathsf{assigns}(S)$ observationally.
\lc{Holds for resolved lvalues only; note aliasing writes are handled by Proposition~\ref{prop:ptr}.} \weiqinote{Addressed: the proof now notes that this direction holds for resolved lvalues, with aliased writes handled by Proposition~\ref{prop:ptr}.}

($\not\Leftarrow$) Let $h \notin S$ and consider $f$ that executes $h := 99;\; h := 0$. At function exit $\sigma_{\mathit{post}}(h) = 0 = \sigma_{\mathit{pre}}(h)$, so $f$ satisfies observational $\mathsf{assigns}(S)$. However, $(h := 99) \in \mathit{Writes}(f)$ with $h \notin S$, so $f$ violates operational $\mathsf{assigns}(S)$. Thus the implication is strict.
\end{proof}

The observational encoding is sufficient for sequential programs because the caller observes memory \emph{only at function return}; writes that are undone before return are invisible to the caller and to the postcondition. The following two theorems formalise this and establish that the snapshot-and-assert mechanism is a sound BMC implementation of the classical Hoare frame rule.

\begin{theorem}[Observational Frame Rule]
\label{thm:frame}
Let $\mathit{havoc}(s, S) = \{\,s' \mid \forall v \notin S \cdot s'.\mathit{vars}(v) = s.\mathit{vars}(v)\,\}$. If $\mathcal{P}$ observes only variables outside $S$, i.e.\
\[ \forall s_1, s_2 \cdot \bigl(\forall v \notin S \cdot s_1.\mathit{vars}(v) = s_2.\mathit{vars}(v)\bigr) \Rightarrow \bigl(\mathcal{P}(s_1) \leftrightarrow \mathcal{P}(s_2)\bigr) \]
then $\mathcal{P}(s) \Rightarrow \forall s' \in \mathit{havoc}(s, S) \cdot \mathcal{P}(s')$. \emph{Lean~4}: \texttt{phase3\_havoc\_sound}.
\end{theorem}

\begin{theorem}[Hoare Frame Rule Embedding]
\label{thm:hoare_frame}

Define \[
\mathit{hoare}(P, C, Q) \triangleq \forall s \cdot P(s) \Rightarrow Q(C(s))
\]. $C$ satisfies observational $\mathsf{assigns}(S)$ iff $\forall s, v \notin S \cdot (C\,s).\mathit{vars}(v) = s.\mathit{vars}(v)$ (dot notation denotes record field access, following Lean convention). A frame property $F$ is \emph{independent of}~$S$ iff $\forall s_1, s_2 \cdot (\forall v \notin S \cdot s_1.\mathit{vars}(v) = s_2.\mathit{vars}(v)) \Rightarrow (F(s_1) \leftrightarrow F(s_2))$. Then:
\[ \mathit{hoare}(P,\,C,\,Q) \;\wedge\; \mathsf{assigns}(S) \;\wedge\; F \perp S \;\Longrightarrow\; \mathit{hoare}(P \wedge F,\; C,\; Q \wedge F) \]
\emph{Lean~4}: \texttt{hoare\_frame\_rule}.
\end{theorem}

Theorem~\ref{thm:hoare_frame} establishes that the snapshot-and-assert mechanism is a direct BMC mechanisation of the classical Hoare frame rule side condition $\mathit{FV}(F) \cap \mathit{Mod}(C) = \emptyset$~\cite{reynolds2002separation, o2019separation}, not an ad-hoc optimisation. The proof discharges $Q$ from the Hoare hypothesis, then reconstructs $F$ at the post-state: the assigns hypothesis shows the non-assigned variables are unchanged, and the frame property transfers across that equality.

The machine-checked proofs are in the supplementary material; \texttt{hoare\_frame\_rule} is shown in Figure~\ref{fig:lean_hoare}; it is the most instructive of the six.

\begin{figure}[t]
\begin{lstlisting}[basicstyle=\ttfamily\scriptsize]
theorem hoare_frame_rule
  (P Q F : State -> Prop) (C : State -> State)
  (S vars_frame : List String) :
  hoare P C Q ->                            -- {P} C {Q}
  observational_assigns C S ->              -- C respects assigns(S)
  is_frame F vars_frame ->                  -- F observes vars_frame
  (forall v in vars_frame, v not_in S) ->   -- frame disjoint from S
  hoare (fun s => P s /\ F s) C (fun s => Q s /\ F s) := by
  intro h_hoare h_assigns h_frame h_disj s (h_P, h_F)
  exact (h_hoare s h_P,
         (h_frame s (C s) (fun v hv =>
            (h_assigns s v (h_disj v hv)).symm)).mp h_F)
\end{lstlisting}
\caption{The machine-checked \texttt{hoare\_frame\_rule} (Lean~4, lightly
pretty-printed): the snapshot-and-assert mechanism is a direct realisation of
the Hoare frame rule side condition. $Q$ is discharged via \texttt{h\_hoare};
the frame $F$ is transferred across the post-state using \texttt{h\_assigns}
(non-assigned variables are unchanged).}
\label{fig:lean_hoare}
\end{figure}

\paragraph{Pointer disjunction extension.}
The implementation extends the scalar model to pointer-typed assigns targets via a disjunctive assertion. For \texttt{assigns(*ptr)}, the module generates for each non-assigned global,~$h$, the assertion $\mathtt{assert}(h = \mathit{snap}_h \;\vee\; \mathit{ptr} = \&h)$, motivated by the need to handle output-channel pointer patterns such as \texttt{assigns(*hdr\_len, *data\_len)} in the X.509 case study. \marienote{this is the first mention of the X.509 case study and it comes without any context, wither remove it from here or give it more description (I prefer removing it at this point)}

\begin{proposition}[Pointer Disjunction Soundness]
\label{prop:ptr}
Let \textup{ptr} be a pointer parameter whose value is not modified during the function body, and let $h$ be a global with $\&h \neq \mathit{ptr}$ at function entry. If the assertion $h = \mathit{snap}_h \vee \mathit{ptr} = \&h$ holds at function exit, then $h = \mathit{snap}_h$, i.e.\ $h$ satisfies the observational assigns condition.
\end{proposition}
\begin{proof}
Since \textup{ptr} is not modified, $\mathit{ptr}_{\mathit{exit}} = \mathit{ptr}_{\mathit{entry}}$. By assumption $\mathit{ptr}_{\mathit{entry}} \neq \&h$, so $\mathit{ptr}_{\mathit{exit}} \neq \&h$. The right disjunct is false; thus the left disjunct $h = \mathit{snap}_h$ must hold.
\end{proof}

The restriction to unmodified pointer values covers the common case of pointer parameters used as output channels. If the pointer itself is reassigned within the body, the disjunction is conservative (may generate false alarms) but remains sound: passing the assertion still implies that the observational condition holds for the at-entry binding of \textup{ptr}. Proposition~\ref{prop:ptr} and the full frame-compliance result are machine-checked alongside the five scalar theorems.

\subsection{When Operational Semantics Is Required}
\label{sec:scope}

The observational encoding is sound for \emph{sequential, single-threaded C programs} with no signal handlers accessing shared globals. Three program classes require the stronger operational semantics:

\begin{enumerate}
  \item \textbf{Concurrent programs.} Another thread may read shared memory between the write $h := 99$ and the restore $h := 0$; the intermediate state is observable.
  \item \textbf{Signal and interrupt handlers.} An interrupt handler may observe shared globals at any point during function execution; the same argument applies.
  \item \textbf{Cache-partitioning and WCET analysis.} These analyses bound physical write operations, not just the final state.
\end{enumerate}

These constraints define the boundary of sequential BMC itself. Any analysis
that requires observing intermediate writes falls outside the sequential,
single-threaded execution model that BMC targets, so these are limitations of
the setting rather than of our approach.

Two further boundaries are specific to our \emph{implementation}. First, the observational frame check tracks non-pointer objects of static lifetime, globals and function-local statics alike, arrays element by element up to 256 elements and as one equality beyond that; writes through pointer-typed globals and to the heap are outside its snapshot set, so their frame compliance is not enforced (\S\ref{sec:threats}). Second, a function taking an unbounded-size pointer input needs an \texttt{is\_fresh} extent precondition, which \toolname{} does not currently discharge on the Replace side: the check reduces there to a validity predicate alone, which is unconditionally false when the caller's argument has automatic storage and which constrains neither the stated extent nor separation between parameters (esbmc/esbmc\#6542, an implementation defect rather than a semantic tension). Our case studies avoid the first boundary by construction and the second by stating pointer preconditions as \texttt{p != NULL}, which on the current implementation concedes nothing that \texttt{is\_fresh} would deliver at a Replace site (\S\ref{sec:threats}); lifting both is future work. \marienote{I suggest deleting the previous sentence, this is one of those forward references that is not particularly helpful at this point}

\section{Implementation in \toolname{}}
\label{sec:methodology}

\toolname{} realises the contract framework of Section~\ref{sec:semantics} as a
single \texttt{frame\_enforcer} module that drives both verification phases
through shared snapshot/assert infrastructure (Figure~\ref{fig:frame_enforcer}).
The module exposes four operations: \texttt{materialize\_snapshots()} emits
\textsc{decl}\,+\,\textsc{assign} IR recording the entry value of each tracked
variable; \texttt{enforce\_frame\_rule()} emits, for every non-assigned
variable, an \textsc{assert} against its snapshot in Enforce mode or an
\textsc{assume} in Replace mode; \texttt{classify\_assigns\_targets()} separates
direct from pointer-typed targets so that aliasing disjunctions can be generated;
and \texttt{collect\_global\_variables()} enumerates the global scalars from the
symbol table, excluding pointer-typed and compiler-internal symbols. Because the
same snapshot machinery underlies both phases, the assert/assume duality is the
only difference between enforcing a contract and abstracting a call site.
\emph{Scope of the frame check.} This snapshot mechanism covers non-pointer
objects of static lifetime, globals and function-local statics alike, with
arrays held element by element up to 256 elements and as one equality beyond;
pointer-typed globals and heap cells are not snapshotted, so a write to such a
location outside the assigns clause is not caught by Enforce (the pointer disjunction of Proposition~\ref{prop:ptr}
handles \texttt{assigns(*ptr)} \emph{targets}, not the exclusion of non-assigned
pointer globals from the tracked set). DFCC's per-write monitoring, by contrast,
covers these locations; we treat this as a soundness-scope limitation
(\S\ref{sec:threats}).

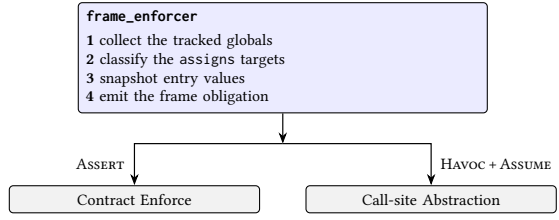
\begin{figure}[t]
\centering
\begin{tikzpicture}[
  font=\scriptsize,
  box/.style={draw,rounded corners=1.5pt,inner sep=3pt,line width=.4pt,align=center},
  ctx/.style={box,fill=black!5,text width=31mm},
  ar/.style={-{Stealth[length=1.5mm]},line width=.45pt},
  elbl/.style={font=\scriptsize\scshape,inner sep=1.2pt,fill=white},
]
\node[box,fill=blue!8,text width=52mm,align=left] (mod) {%
  \textbf{\texttt{frame\_enforcer}}\\[2pt]
  \textbf{1} collect the tracked globals\\
  \textbf{2} classify the \texttt{assigns} targets\\
  \textbf{3} snapshot entry values\\
  \textbf{4} emit the frame obligation};
\coordinate (s) at ($(mod.south)+(0,-4mm)$);
\draw[ar] (mod.south) -- (s);
\node[ctx,anchor=north east] (ce) at ($(s)+(-3mm,-5.5mm)$) {Contract Enforce};
\node[ctx,anchor=north west] (ca) at ($(s)+(3mm,-5.5mm)$) {Call-site Abstraction};
\draw[ar] (s) -| (ce.north);
\draw[ar] (s) -| (ca.north);
\node[elbl,anchor=east] at ($(ce.north)+(-0.8mm,2.8mm)$) {Assert};
\node[elbl,anchor=west] at ($(ca.north)+(0.8mm,2.8mm)$) {Havoc\,+\,Assume};
\end{tikzpicture}
\Description{The frame_enforcer module: its four operations, and the two
contexts they are instantiated in.}
\caption{The \texttt{frame\_enforcer} module. Steps~1--3 are shared by both
contexts, so the \textsc{assert}/\textsc{assume} choice in step~4 is the
only difference between checking a contract and abstracting a call site.
Step~1 tracks non-pointer objects of static lifetime.}
\label{fig:frame_enforcer}
\end{figure}
\lc{Label the four module operations on the figure; caption should name inputs, outputs, workflow.} \weiqinote{Addressed: the figure is redrawn (in TikZ, so it stays editable). The four operations are now labelled and numbered in execution order, and the caption states what the module consumes and emits. We label them \emph{semantically} (``collect the tracked globals'', ``snapshot entry values'') rather than by their C\texttt{++} entry-point names, which are implementation detail a reader does not need to follow the architecture; the names themselves remain in the \S\ref{sec:methodology} prose. Redrawing also lets the figure carry what the previous raster could not: that steps~1--3 are shared, so the \textsc{assert}/\textsc{assume} choice in step~4 is the only difference between the two contexts.}

\subsection{Enforce Mode: Snapshot-and-Assert}

Given \texttt{--enforce-contract~$f$} (assigns compliance is checked by
default), \toolname{} generates the wrapper of Figure~\ref{fig:enforce}: it
snapshots the non-assigned globals, assumes the precondition, calls the (renamed)
original body, and asserts both assigns compliance and the postcondition. The frame assertions
(lines~8--9) fire only at function exit, never mid-execution, which realises
observational semantics by construction and yields $(G{-}M)+E$ VCCs. For a
pointer target \texttt{assigns(*ptr)}, the module emits a \emph{disjunctive}
guard for each non-assigned global~$h$,
$\mathtt{assert}(h = \mathit{snap}_h \vee \mathit{ptr} = \&h)$, which is sound:
if \texttt{ptr} does not alias~$h$, then $h$ is unchanged (Proposition~\ref{prop:ptr}).

\begin{figure}[t]
\begin{lstlisting}[language=C]
void compute_wrapper(int *arr, int n,
                     int lo, int hi) {
  /* 1. Snapshot every NON-assigned global (g_processed,
        the assigns target, is exempt -- it may change) */
  int snap_other = g_other;

  /* 2. Assume precondition */
  __ESBMC_assume(
      __ESBMC_is_fresh(arr, n*sizeof(int))
      && n > 0 && lo >= 0 && lo <= hi
      && hi <= 100);

  /* 3. Call original (renamed) body */
  int rv = __ESBMC_contracts_original_compute(
               arr, n, lo, hi);

  /* 4. Frame compliance: ASSERT at exit */
  assert(g_other == snap_other); /* per non-assigned global */
  /* For pointer target *ptr: disjunctive guard:
     assert(g == snap_g || ptr == &g)            */

  /* 5. Postcondition */
  assert(rv >= n * lo && rv <= n * hi);
}
\end{lstlisting}
\caption{Enforce wrapper generated for \texttt{compute()} (pseudo-C rendering of GOTO IR).}
\label{fig:enforce}
\end{figure}

\subsection{Replace Mode: Selective Havoc at Call Sites}

At each call site the body is never examined: \toolname{} replaces the entire
call (Figure~\ref{fig:replace}, lines~1--4) with an assertion of the
precondition, a havoc of the declared assigns targets, and an assumption of the
postcondition, so no loop unrolling occurs, which is why Replace eliminates the
$O(N \times K)$ scaling. Havocking follows a three-tier strategy: \emph{precise}
(explicit assigns list, only listed targets havoced); \emph{empty}
(\texttt{\_\_ESBMC\_assigns(0)}, a pure function, no havoc); and
\emph{conservative} (no assigns clause, all global scalars havoced).
Field-sensitive targets such as \texttt{assigns(*p)} or \texttt{assigns(arr[i])}
are handled in both modes by AST-level substitution of formal parameters with
actual arguments, requiring no alias analysis.

\begin{figure}[t]
\begin{lstlisting}[language=C]
/* 1. Caller must satisfy precondition */
assert(__ESBMC_is_fresh(buf, n*sizeof(int))
       && n > 0 && ...);

/* 2. Havoc declared assigns targets only */
g_processed = nondet_int();

/* 3. Nondeterministic return value */
int result = nondet_int();

/* 4. Assume postcondition */
__ESBMC_assume(result >= n * lo
               && result <= n * hi);
/* original call removed */
\end{lstlisting}
\caption{Call-site transformation for Replace mode.}
\Description{Pseudocode showing call-site transformation for Replace mode.}
\label{fig:replace}
\end{figure}

\marienote{to improve the flow try to end each section with a summary sentence or two: this is what we presented and this is what we will present next ...}
\section{Analytical Model}
\label{sec:model}

\subsection{VCC Complexity}

Table~\ref{tab:vcc-complexity} shows the VCC count formulas. The observational formula $(G-M)+E$ follows directly from the snapshot-and-assert construction: one assertion per non-assigned global plus one per postcondition. The operational formula $20+W+E$ arises from DFCC's per-write membership tests (one per write instruction) plus one per postcondition, plus a fixed infrastructure offset of~20 VCCs measured empirically on \baseline{}~6.8.0.


\begin{table}[t]
\small
\caption{VCC complexity. $G$: globals, $M$: assigns targets, $W$: write instructions, $E$: ensures clauses.}
\label{tab:vcc-complexity}
\begin{tabular}{llll}
\toprule
\textbf{Approach} & \textbf{VCCs} & \textbf{Scales with} & \textbf{Indep.\ of} \\
\midrule
\toolname{} (Observational) & $(G-M)+E$ & $G$ (globals) & $W$ (writes) \\
\baseline{} (DFCC)          & $20+W+E$  & $W$ (writes)  & $G$ (globals) \\
\bottomrule
\end{tabular}
\end{table}

\paragraph{On the constant 20.}
The baseline offset of~20 VCCs for \baseline{} (DFCC) is the intercept of its linear VCC growth, measured empirically on \baseline{}~6.8.0: a single-function harness with $W$ write instructions and $E$ postconditions produces $20+W+E$ VCCs for $W \ge 1$ (the offset arises from internal infrastructure checks inserted by \texttt{goto-instrument}: null-pointer guards, frame overhead, etc.; the empty-body $W{=}0$ case sits off this line, as the write-set machinery is partly elided). \marienote{the text inside the () might be better as a footnote} This constant is version-specific; all experiments use \baseline{}~6.8.0, for which $20+W+E$ matched every test harness (Table~\ref{tab:rq1}).

\paragraph{Crossover condition.}
\toolname{} generates fewer VCCs than \baseline{} when:
\[
  (G - M) + E < 20\mathbin{\text{+}}W\mathbin{\text{+}}E
  \;\;\Longleftrightarrow\;\;
  G < 20\mathbin{\text{+}}M\mathbin{\text{+}}W
\]
The observational approach pays a fixed cost in globals~$G$, and the
operational approach pays per write~$W$. The crossover is where the global
count equals the write count plus the DFCC infrastructure overhead. This
counts verification conditions, not solver work: a VCC is not a comparable unit
across \toolname{}'s SMT and \baseline{}'s SAT backends, so the criterion says
which encoding emits fewer obligations, not which tool finishes first. The empirical results in Exp\,1b (one of the synthetic experiments described in Section~\ref{sec:experiments}; ``Exp'' labels are used throughout to identify individual experiment configurations) confirm the crossover: at $G=22$, $M=1$, $W=1$, the VCC counts are equal (consistent with $G = 20\mathbin{\text{+}}M\mathbin{\text{+}}W$). Figure~\ref{fig:crossover} visualises the full $(G, W)$ parameter space analytically: the green region ($G < 20{+}M{+}W$) is where observational encoding generates fewer VCCs; the red region is where DFCC wins; the diagonal dashed boundary is the exact crossover $G = 20{+}M{+}W$, and the empirical point from Exp\,1b is marked.
\lc{List the 21 checks explicitly; ``etc.'' blocks reproducibility of a load bearing constant.} \weiqinote{Partially addressed, and the target moved: the constant is 20, not 21 (see the operational-encoding paragraph in \S\ref{sec:background}). The paragraph above now reports the \emph{measurement} that produces it, a sweep over $W$ and $E$ on \baseline{}~6.8.0, raw table in the artifact, plus the off-line $W{=}0$ special case, rather than a prose list ending in ``etc.''. I did not enumerate the individual checks by name; say the word if you still want that.} \marienote{I suggest adding the list of 21 checks to an appendix or a file in the repo for this paper.}


\begin{figure}[t]
  \centering
  \includegraphics[width=1\columnwidth]{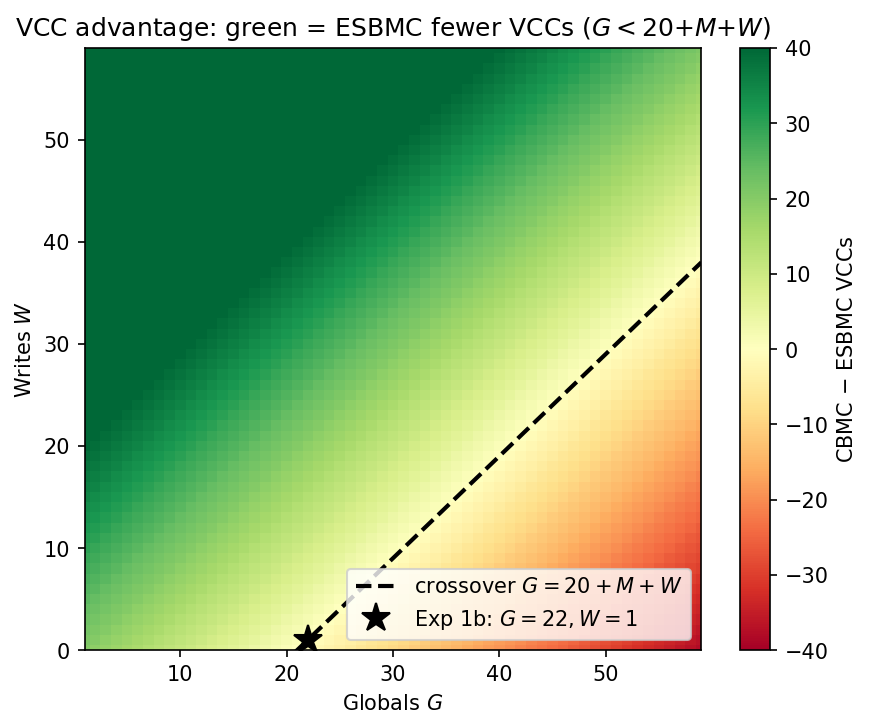}
  \caption{VCC advantage heatmap over $(G, W)$ space ($M=1$, $E=1$). Green = \toolname{} wins ($\ensuremath{G < 20\text{+}M\text{+}W}$); red = \baseline{} wins; diagonal dashed line = crossover $\ensuremath{G = 20\text{+}M\text{+}W}$. The boundary is derived analytically from the VCC formulas; the empirical crossover at $\ensuremath{G{=}22}$ (Exp\,1b) is marked.}
  \label{fig:crossover}
\end{figure}
\lc{Crossover is in $G$, not $N$; fix caption/axis notation to match the text.} \weiqinote{Addressed: the figure is regenerated with axes in $G$ and $W$, the dashed boundary labelled $G = 20{+}M{+}W$, and the Exp\,1b empirical point starred at $G{=}22$. Caption and axis notation now match the text.}

\subsection{Breakeven Call Count \texorpdfstring{$K^*$}{K*}}

Modular verification incurs a one-time Enforce cost, $T_E$, and per-call Replace cost, $T_R$; direct verification costs $T_D$ per call. Modular pays off when $K > K^*$, where:
\[ K^* = \frac{T_E}{T_D - T_R} \]

Table~\ref{tab:kstar} shows the measured values for \texttt{compute()} at varying loop bound~$N$. Nothing in the definition confines $K^*$ to values above one. When $T_E \ll T_D$, that is, when verifying a function against its precondition alone is much cheaper than verifying the same body inside its caller, $K^*$ falls below~1 and a single call site already pays. Section~\ref{sec:svcomp} measures four such cases, with $K^*$ between $0.05$ and $0.27$. The ratio is only as meaningful as the property $T_D$ is measured against: an under-specified harness deflates $T_D$ and inflates $K^*$, which is how those four were first misread as sitting below crossover.


\begin{table}[t]
\small
\caption{Breakeven call count $K^*$ ($T_D{=}$Direct at $K{=}1$, $T_E{=}$Enforce, $T_R{=}$Replace). Modular pays off once $K > K^*$.}
\label{tab:kstar}
\begin{tabular}{rrrrrc}
\toprule
$N$ & $T_D$ (ms) & $T_E$ (ms) & $T_R$ (ms) & $K^*$ & $K{=}1$ speedup \\
\midrule
 5 & 200  & 367   & 6 & 1.89 & 0.5$\times$ \\
10 & 597  & 1787  & 5 & 3.02 & 0.3$\times$ \\
20 & 2771 & 12503 & 5 & 4.52 & 0.2$\times$ \\
\bottomrule
\end{tabular}
\end{table}

For \texttt{compute()} the one-time Enforce cost $T_E$ must be amortised over the call sites: modular verification is slower than direct at $K{=}1$ and pays off only once $K > K^*$. Here $K^* \approx 2$--$5$, increasing with the loop bound~$N$, because for this tight numeric postcondition the isolated Enforce proof itself grows with~$N$ (from 367\,ms at $N{=}5$ to 12.5\,s at $N{=}20$), faster than a single direct call $T_D$. The Replace cost is negligible ($T_R < 20$\,ms, a pure havoc). For loop-free functions $T_D \approx T_R$ so $K^* \to \infty$ and direct inlining is always competitive, as confirmed by the X.509 case study (Section~\ref{sec:casestudy}). The decisive regime is large $N{\times}K$, where direct verification times out entirely while modular completes (Table~\ref{tab:rq2}).

\section{Experimental Evaluation}
\label{sec:experiments}

\paragraph{Setup.}
All experiments run on x86\_64 Linux (Ubuntu~24.04) on an idle 16-vCPU Intel Xeon 6982P-C with 61\,GiB RAM, governor pinned to \texttt{performance}, clock steady between 3.5 and 3.9\,GHz. \toolname{}: \textbf{one build for every measurement}, \texttt{master} commit \texttt{2f719f59e0}, which reports version 8.5.0 and lies 61 commits after the 8.5 release tag. The build was configured without the ibex interval library. The 8.4 and 8.5 Linux releases include it, and it leaves the FPU rounding mode upward at startup (\toolname{} issue \#8043), after which Bitwuzla exhausts memory on Direct baselines that take seconds here. The 8.5 release also lacks five \texttt{is\_fresh} fixes from that span and fails X.509 scenario~B, and the later \texttt{v8.6} line changes the Enforce harness again (\#8042) and fails the ML-KEM contracts. The artifact therefore builds the measured commit without ibex, which reproduces every verdict here and the timings within a few percent on the same CPU model. Default solver Bitwuzla~\cite{niemetz2023bitwuzla}; the artifact records that commit, the binary's SHA-256 and the clock per measurement. Runs are serial, so concurrent processes never contend for cores. \baseline{}: version 6.8.0, DFCC via \texttt{goto-cc} $\to$ \texttt{goto-instrument} $\to$ \texttt{cbmc}. Timeout 120\,s (600\,s for the $N{=}256$ scaling study, Table~\ref{tab:mlkem-scaling}). \textbf{All timing tables report medians of five runs}, with two stated exceptions: the FreeRTOS \texttt{heap\_1} checks (\S\ref{sec:freertos}) are single-run, and the real-size scaling table (Table~\ref{tab:mlkem-scaling}, $N \in \{64,128,256\}$) uses three repetitions for Enforce and Replace and a single run for Direct, because one Direct run alone costs $182$\,s at $N{=}64$, $495$\,s at $N{=}128$, and exhausts the $600$\,s budget at $N{=}256$. The Direct entries there are consequently single measurements; the two that completed agree with the monotone trend, and the $N{=}256$ conclusion is a timeout in both tools rather than a time comparison, so it does not rest on run-to-run precision. Run-to-run variation is below~5\,ms for sub-second cases and below~2\% for multi-second cases (min--max spread is in the artifact). \textbf{Units.} Cross-tool comparisons (ML-KEM Tables~\ref{tab:mlkem-esbmc}--\ref{tab:mlkem-cbmc}, scaling Table~\ref{tab:mlkem-scaling}) use \emph{end-to-end wall-clock for both tools}: including \baseline{}'s full \texttt{goto-cc}$+$\texttt{goto-instrument}$+$\texttt{cbmc} pipeline, so the two tools are measured in the same unit. Within-\toolname{} synthetic experiments (RQ1--RQ2, Tables~\ref{tab:rq1}--\ref{tab:rq2}; SV-COMP, Table~\ref{tab:svcomp}) report \toolname{}'s \emph{solver} (BMC) time, isolating verification work from per-invocation process startup; for \baseline{} in RQ1 we likewise report solver time (the \texttt{goto-cc}$+$\texttt{goto-instrument} steps add $80$--$135$\,ms, rising with~$W$ as the DFCC pass instruments each write). The X.509 case study reports full-pipeline wall-clock for both tools.
\lc{Add a benchmark overview table; give suite, count, origin, selection criteria, availability, before RQ1.} \weiqinote{Addressed: Table~\ref{tab:suites} is a new benchmark overview, suite, subject count, provenance, selection criterion, and what each exercises, placed before RQ1.}
\lc{Headline RQ2/ML-KEM results are single run; report medians plus variance for the multi second cases.} \weiqinote{Addressed, and this triggered a full re-measurement. All \toolname{} timings were re-collected on one build, run serially on an idle machine, as medians of five runs with min--max spread in the artifact. That build has since been superseded: it was described as a frozen 8.4.0 build but its commit was never recorded, so the numbers could not be checked against the tool that produced them. Every measured table is now re-collected on \toolname{} at \texttt{master} commit \texttt{2f719f59e0}, on a dedicated idle server, and \texttt{check/verify\_paper.py} in the artifact re-derives each figure from the logs. Two measurements are not five-run and are now flagged as such in the setup paragraph: the FreeRTOS check (single-run) and the real-size scaling table (three repetitions, single Direct run per~$N$, since one Direct run costs 182\,s, 495\,s, or times out). It mattered: the old ``$4\times$ cheaper than DFCC'' was a single-run outlier. A later full re-baseline on one pinned commit puts the honest value at $2.3$--$2.8\times$ and shows it does not grow monotonically with~$N$; both are now reported as measured.}



We address three research questions (RQ1--RQ3) testing the analytical model, followed by four case studies. Table~\ref{tab:suites} summarises the benchmark suites, their provenance, and why each is included; author-written subjects are used only for the synthetic model-testing experiments, while every case study uses third-party code with, where available, third-party contracts.

\begin{table*}[t]
\footnotesize
\setlength{\tabcolsep}{4pt}
\caption{Benchmark suites: provenance and selection. ``Author'' = written by us to vary one model parameter in isolation; ``third-party'' = external code/specs we did not author.}
\label{tab:suites}
\begin{tabular}{lllp{4.6cm}p{4.4cm}}
\toprule
Suite & \#\,Subjects & Provenance & Selection criterion & What it exercises \\
\midrule
Synthetic RQ1 (VCC/$W$) & 2 sweeps ($W,G$) & Author & vary write count $W$ (Exp\,1a) and global count $G$ (Exp\,1b) in isolation & assigns-cost decoupling from body complexity; the VCC crossover $G{=}20{+}M{+}W$ \\
Synthetic RQ2 (\texttt{compute}) & $3{\times}4$ grid & Author & vary loop bound $N$ and call count $K$ & call-count amortisation and the breakeven $K^*$ \\
RQ3 boundary patterns & 6 patterns & Author & one idiom per field-/index-sensitive assigns pattern & semantic-boundary detection (not a coverage claim) \\
Seeded faults (\S\ref{sec:ce}) & 6 mutants & Author & one mutant per fault class on the running example & relative soundness: no Direct-fail/Modular-pass \\
SV-COMP VerifyThis & 6 tasks & Third-party & VerifyThis tasks that compile under bounded arrays and admit a sound bounded contract & single-call-site ($K{=}1$) below-crossover regime \\
ANSSI X.509 & 1 parser, 7 scenarios & Third-party & the ANSSI reference X.509 parser & loop-free, pointer-heavy, multi-level call chains ($K^*{\to}\infty$) \\
FreeRTOS \texttt{heap\_1} & 1 function, 9 checks & Third-party & the \texttt{pvPortMalloc} allocator & pointer/heap allocator logic in sequential isolation \\
ML-KEM (\texttt{mlkem-native}) & 3 kernels & Third-party code + contracts & kernels spanning a range of per-coefficient body cost & scaling on FIPS-203 code to the real size $N{=}256$ \\
\bottomrule
\end{tabular}
\end{table*}


\subsection{RQ1: Does VCC Count Decouple from Body Complexity?}
\label{sec:rq1}

\paragraph{Hypothesis.}
\toolname{}'s VCC count is constant as $W$ increases; \baseline{}'s grows linearly.

\paragraph{Setup.}

Exp\,1a: $G=5$ globals, $M=1$ assigns target, $E=1$ ensures, $W \in \{1, 10, 50, 100, 200, 500\}$ sequential writes.

Exp\,1b (complementary): fixed $W=1$. We vary 
\\$G \in \{1, 5, 10, 20,
50, 100, 200, 500\}$ to test the $G$-scaling direction.

\paragraph{Results.}
Table~\ref{tab:rq1} confirms both formulas exactly, and the mechanism is worth
stating because it is what the rest of the paper rests on. The observational
check asks one question per global that the clause did not name: was this
location the same at exit as at entry? The body's writes do not enter that
question. Whether \texttt{compute} assigns \texttt{g0} once or five hundred
times, the wrapper snapshots the same four unnamed globals and asserts the same
four equalities, so \toolname{} emits $(G{-}M){+}E = 5$ verification conditions
at every~$W$. DFCC asks the dual question, one per write: is \emph{this} write
inside the write set? That count is a property of the body, so it is
$20{+}W{+}E$, here $21{+}W$.

The consequence is a ratio that grows without bound in~$W$: $4.4\times$ at
$W{=}1$ and $104.2\times$ at $W{=}500$. This is the result that survives every
change of setting, because it is a \emph{count}, not a time. It is fixed by the encoding
before any solver runs, so it does not move with the machine, the SMT backend,
or the solver's version, and re-measuring the whole study on different hardware
and a different \toolname{} build reproduced every cell of both columns
unchanged while every timing in the paper moved. Time follows the count only
loosely: at $W{=}500$, $104\times$ fewer verification conditions buy $6.2\times$
less wall-clock, because the conditions are not equally hard and because
per-invocation start-up, $106$--$133$\,ms here, dominates a small body.

The decoupling is directional, not free. Exp\,1b holds $W{=}1$ and varies~$G$:
\baseline{} stays at 22\,VCCs while \toolname{} grows as $(G{-}M){+}E = G$,
crossing over at $G{=}22$, exactly where $G < 20{+}M{+}W$ stops holding. So the
observational encoding trades a dependence on the body's write count for a
dependence on the number of globals in scope, which is the right trade whenever
functions are larger than the state they may modify. That is the common case in
the code we study: \texttt{mlk\_poly\_reduce} writes 256 coefficients and touches
no global at all. Figure~\ref{fig:vcc_scaling} shows both directions.


\begin{table}[t]
\small
\caption{RQ1: VCC count and time vs.\ write count $W$ ($G=5$, $M=1$, $E=1$).}
\label{tab:rq1}
\begin{tabular}{rrrrrr}
\toprule
\multirow{2}{*}{$W$}
  & \multicolumn{2}{c}{\toolname{} (Obs.)}
  & \multicolumn{2}{c}{\baseline{} (DFCC)}
  & \multirow{2}{*}{VCC ratio} \\
\cmidrule(lr){2-3}\cmidrule(lr){4-5}
  & VCCs & Time & VCCs & Time & \\
\midrule
1   & 5 & 106\,ms & 22  & 41\,ms  & 4.4$\times$ \\
10  & 5 & 110\,ms & 31  & 55\,ms  & 6.2$\times$ \\
50  & 5 & 110\,ms & 71  & 116\,ms & 14.2$\times$ \\
100 & 5 & 112\,ms & 121 & 196\,ms & 24.2$\times$ \\
200 & 5 & 118\,ms & 221 & 353\,ms & 44.2$\times$ \\
500 & 5 & 133\,ms & 521 & 823\,ms & 104.2$\times$ \\
\bottomrule
\end{tabular}
\end{table}

\begin{figure*}[t]
  \centering
  \begin{minipage}{0.4\textwidth}
    \centering
    \includegraphics[width=\linewidth]{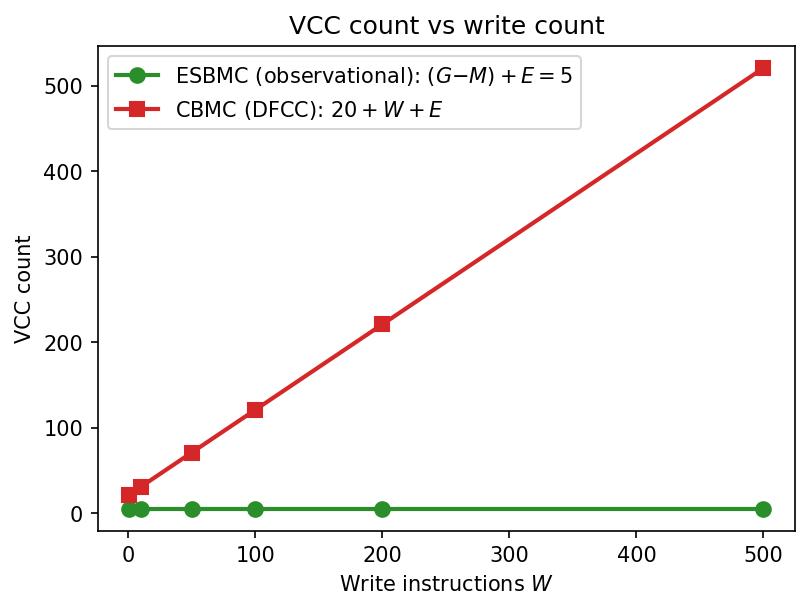}
  \end{minipage}\hfill
  \begin{minipage}{0.4\textwidth}
    \centering
    \includegraphics[width=\linewidth]{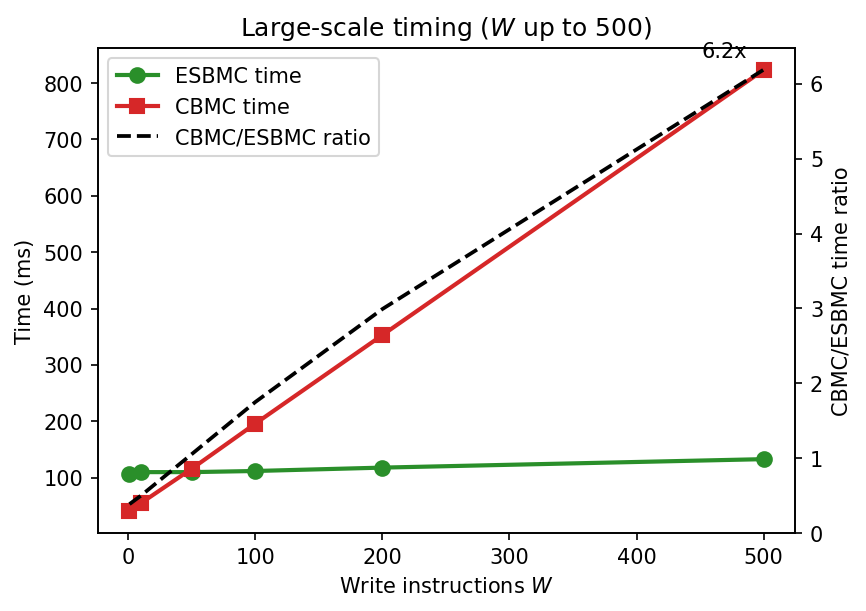}
  \end{minipage}
  \caption{Left: VCC counts vs.\ $W$ ($G=5$, $M=1$, $E=1$). \toolname{} is constant at $(G{-}M)+E=5$, while \baseline{} grows linearly as $20+W+E$. Right: large-scale timing ($W$ up to 500); the dashed line (right axis) gives the \baseline{}/\toolname{} speedup ratio, reaching $6.2{\times}$ at $W=500$.}
  \label{fig:vcc_scaling}
\end{figure*}

\begin{dialogbox}
\paragraph{Answer to RQ1.}
Yes. \toolname{}'s VCC count is fully decoupled from $W$. At $W=500$: $104{\times}$ fewer VCCs, but only $6.2{\times}$ faster, a gap that itself shows VCC \emph{count} is not a linear proxy for solver time (the two scale differently), so we treat the VCC formulas as a structural characterisation and rely on the direct timing experiments for performance claims. \baseline{} times here cover the \texttt{cbmc} step only (a lower bound); the full three-step pipeline adds $80$--$135$\,ms of overhead orthogonal to the VCC decoupling.
\end{dialogbox}
\lc{VCC predicts crossover sign, not magnitude; 52$\times$ VCCs give 5.5$\times$ time, so add the scatter backing $r$.} \weiqinote{Addressed by removal: the Pearson $r > 0.97$ claim is gone, it was overreach. The tension is now stated directly in both RQ1 and \S\ref{sec:threats}: at $W{=}500$ \baseline{} emits $104{\times}$ more VCCs yet runs only $6.2{\times}$ slower, so VCC count is a structural characterisation of encoding cost and explicitly not a linear proxy for solver time. The scatter plot is not drawn; the prose concession seemed to do more work than a figure would, but I will add it if you disagree.}

\subsection{RQ2: When Does Modular Verification Pay Off?}
\label{sec:rq2}


\paragraph{Hypothesis.}
Modular verification pays off once the call count exceeds the crossover~$K^*$, and the speedup grows with~$K$; it is not beneficial at $K{=}1$, where the one-time Enforce cost is unamortised.

\paragraph{Setup.}
Exp\,2: \texttt{compute()} called $K$ times, $N \in \{5, 10, 20\}$, $K \in \{1, 3, 5, 10\}$. We restrict $N \le 20$ because the isolated Enforce cost grows superlinearly in~$N$. The $T_E$ column of Table~\ref{tab:rq2} rises $34\times$ for a $4\times$ increase in~$N$, so beyond $N{=}20$ a single Enforce consumes the whole budget and the $K$ axis this experiment varies is no longer visible. The limit is cost, not expressiveness. The contract states the parameter's extent with \texttt{\_\_ESBMC\_is\_fresh(arr, N*sizeof(int))}, which \toolname{} requires wherever a contract dereferences a pointer parameter; it verifies at $N{=}26$, $40$ and $64$ at no extra cost. The fixed-size ML-KEM study (Section~\ref{sec:mlkem}) exercises larger bodies.

\paragraph{Results.}
Direct grows \emph{superlinearly} in both $N$ and $K$ (Table~\ref{tab:rq2}): at $N=20$, $K=10$ it exceeds the 120\,s timeout. The modular cost is dominated by the one-time Enforce ($T_E$), since Replace is a pure havoc costing under 20\,ms. At $K{=}1$ modular is slower than direct ($0.2$--$0.5{\times}$). It crosses over at $K \approx 2$--$5$ ($=K^*$, Table~\ref{tab:kstar}) and reaches $23{\times}$ at $K{=}10$, beyond which Direct times out.
\begin{table}[th]
\small
\caption{RQ2: Direct vs.\ modular \emph{solver} (BMC) time. ``>120\,s''~=~timeout. Modular~=~Enforce~$+$~Replace; speedup~=~Direct/Modular. Replace verifies the $K$-call harness with \texttt{compute} abstracted; its cost tracks the harness assertion (here trivially discharged by the assumed postcondition, hence $\approx 3$\,ms), not the callee body, so across the paper Replace ranges from a few ms (this harness) to seconds (a quantified postcondition, \S\ref{sec:mlkem}).}
\label{tab:rq2}
\begin{tabular}{rrrrrrr}
\toprule
$N$ & $K$ & Direct & Enforce & Replace & Modular & Speedup \\
\midrule
5  & 1  & 200\,ms & 367\,ms   & 6\,ms  & 373\,ms   & 0.5$\times$ \\
5  & 3  & 604\,ms & 367\,ms   & 13\,ms & 380\,ms   & 1.6$\times$ \\
5  & 5  & 1.2\,s  & 367\,ms   & 14\,ms & 381\,ms   & 3.1$\times$ \\
5  & 10 & 6.6\,s  & 367\,ms   & 18\,ms & 385\,ms   & 17.1$\times$ \\
10 & 1  & 597\,ms & 1787\,ms  & 5\,ms  & 1792\,ms  & 0.3$\times$ \\
10 & 3  & 2.7\,s  & 1787\,ms  & 13\,ms & 1800\,ms  & 1.5$\times$ \\
10 & 5  & 9.9\,s  & 1787\,ms  & 14\,ms & 1801\,ms  & 5.5$\times$ \\
10 & 10 & 42.0\,s & 1787\,ms  & 18\,ms & 1805\,ms  & 23.3$\times$ \\
20 & 1  & 2.8\,s  & 12503\,ms & 5\,ms  & 12508\,ms & 0.2$\times$ \\
20 & 3  & 19.9\,s & 12503\,ms & 13\,ms & 12516\,ms & 1.6$\times$ \\
20 & 5  & 71.3\,s & 12503\,ms & 15\,ms & 12518\,ms & 5.7$\times$ \\
20 & 10 & {>120\,s} & 12503\,ms & 17\,ms & 12520\,ms & {>9.6$\times$} \\
\bottomrule
\end{tabular}
\end{table}

\begin{figure*}[t]
  \centering
  \includegraphics[width=\textwidth]{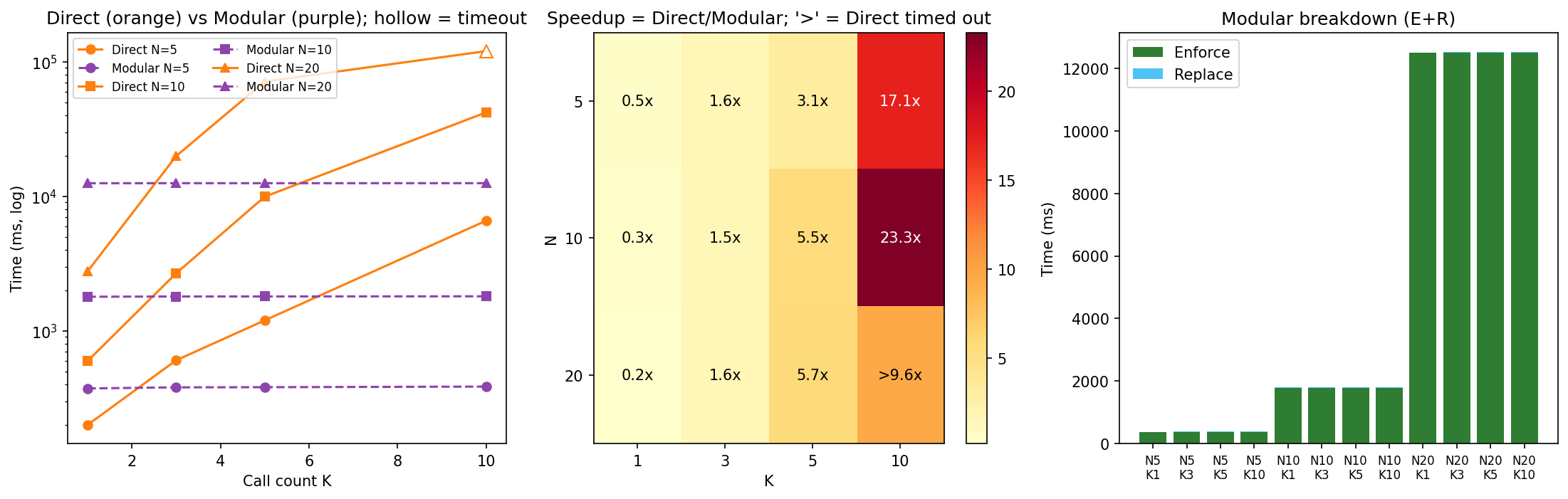}
  \caption{Direct vs.\ modular scaling (Exp\,2). Left: time (log scale): Direct grows as $O(N{\times}K)$ and times out; Modular cost is dominated by the one-time Enforce. Centre: speedup heatmap; cells marked ``\texttt{>}'' are lower bounds where Direct timed out ($>$120\,s). Right: Modular breakdown (Enforce${+}$Replace).}
  \Description{Three-panel figure comparing Direct and Modular verification across loop bound $N$ and call count $K$: log-scale times, a speedup heatmap with timeouts marked, and a cost breakdown showing Enforce plus near-constant Replace.}
  \label{fig:scalability}
\end{figure*}

\begin{dialogbox}
    
\paragraph{Answer to RQ2.}
The decisive advantage is categorical: at $N{=}20$, $K{=}10$ monolithic BMC
times out ($>$120\,s) while the modular pipeline completes in
$12.5$\,s (Figure~\ref{fig:scalability}); it verifies what direct BMC cannot.
Within the budget, the speed-up grows with the call count, reaching $23{\times}$
at $N{=}10$, $K{=}10$, once $K$ passes a crossover $K^* \approx 2$--$5$ that our
$K^*$ model predicts as a bound rather than an exact value
(Table~\ref{tab:kstar}): because direct verification is superlinear in~$N$, the
model pins down the crossover \emph{region} but not a single exact~$K^*$. Below the crossover the
one-time Enforce cost is unamortised ($0.2$--$0.5{\times}$ at $K{=}1$), so the
median over all tested points is $2.4{\times}$. For this benchmark the value of
modular verification lies at higher call counts and, above all, in completing
verification tasks that monolithic BMC cannot. It is not confined to them:
where the isolated Enforce is much cheaper than verifying the same body inside
its caller, $K^*$ drops below~1 and one call site suffices
(\S\ref{sec:svcomp}).
\lc{Model predicts the crossover region, not exactly; Direct is superlinear so $K^*$ is a lower bound.} \weiqinote{Addressed: ``predicts exactly'' is gone; $K^*$ is now stated as a lower bound because Direct is superlinear, and the range is given consistently as ${\approx}2$--$5$ in the conclusion, in RQ2, and in Table~\ref{tab:kstar} (it was 2--3 / 2--4 / 1.86--4.90 in three different places before).}

\end{dialogbox}

\subsection{RQ3: What Is the Scope of Observational Assigns Semantics?}
\label{sec:rq3}

\paragraph{Hypothesis.}
Observational and operational assigns semantics agree on all field-sensitive patterns that arise in practice; they diverge only on the temporary-write pattern (write then restore before exit), which the theoretical analysis in \S\ref{sec:scope} identifies as the boundary case.

\paragraph{Setup.}
Exp\,3: \texttt{foo()} declares \texttt{assigns(g)} but internally writes \texttt{h=99; h=0} (restores before exit), the boundary case. Exp\,3A: six field-sensitive patterns (struct fields, array elements, pointer dereferences, mixed targets).

\paragraph{Results.}
Table~\ref{tab:rq3_boundary} and Figure~\ref{fig:soundness_boundary} show the verdicts. On Exp\,3, observational semantics reports \textsc{pass} (the exit state equals the entry state for all $\ell \notin S$), whereas operational semantics reports \textsc{fail} (the intermediate write to $h \notin S$ is detected). Both encodings agree on all six Exp\,3A patterns.
\lc{Six patterns is a small sample; justify representativeness or expand to strengthen the claim.} \weiqinote{Not addressed. Acknowledged as a limitation rather than fixed: expanding the pattern set is genuinely new work, and I would like your read on whether it earns the space.} \marienote{can you justify only 6 patterns?}

\begin{table}[t]
\small
\caption{RQ3: Semantic boundary.}
\label{tab:rq3_boundary}
\begin{tabular}{llccp{1cm}}
\toprule
& Description & \toolname{} & \baseline{} & Observation \\
\midrule
\multicolumn{5}{l}{\emph{Temporary-write (Exp\,3)}} \\
Exp\,3 & \texttt{assigns(g)}: \texttt{h=99; h=0} &
  \textsc{pass} & \textsc{fail} & Semantics diverge \\
\midrule
\multicolumn{5}{l}{\emph{Field-sensitive (Exp\,3A)}} \\
3A-1 & Struct field \texttt{p->x} & \textsc{pass} & \textsc{pass} & Equivalent \\
3A-2 & Array element \texttt{arr[i]} & \textsc{pass} & \textsc{pass} & Equivalent \\
3A-3 & Multi-field \texttt{m->a, m->c} & \textsc{pass} & \textsc{pass} & Equivalent \\
3A-4 & Global + pointer mixed & \textsc{pass} & \textsc{pass} & Equivalent \\
3A-5 & \texttt{\_\_ESBMC\_old()} postcondition & \textsc{pass} & \textsc{pass} & Equivalent \\
3A-6 & Weak ensures (violation) & \textsc{fail} & \textsc{fail} & Both detect \\
\bottomrule
\end{tabular}
\end{table}

\begin{figure}[t]
  \centering
  \includegraphics[width=0.38\textwidth]{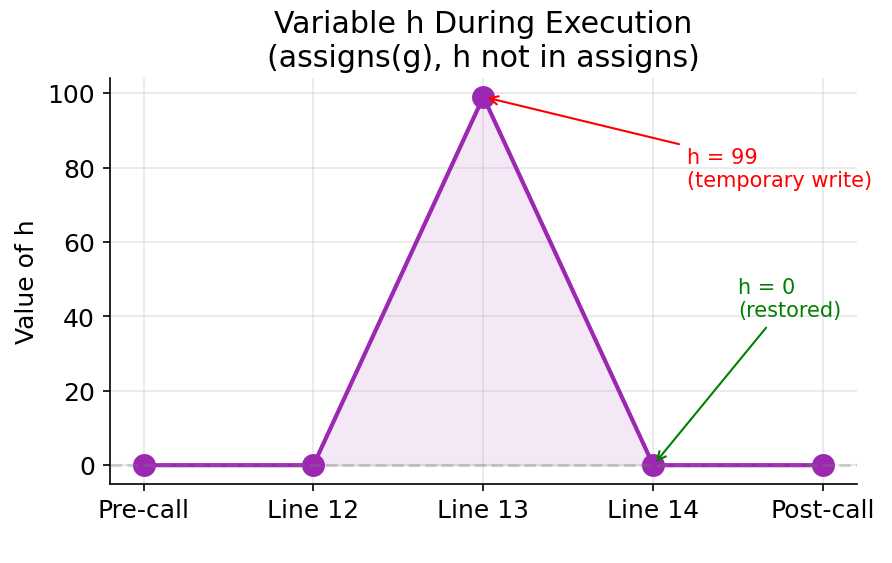}
  \caption{Semantic boundary (Exp\,3), the temporary-write pattern. The
  non-assigned global $h$ spikes to 99 mid-execution and is restored to 0 before
  exit. \toolname{}'s observational rule reports \textsc{pass}
  ($h_\mathit{post}=h_\mathit{pre}$); \baseline{}/DFCC reports \textsc{fail}, as
  the intermediate write to $h$ is intercepted.}
  \label{fig:soundness_boundary}
\end{figure}

\begin{dialogbox}
\paragraph{Answer to RQ3.}
In all tested cases, observational and operational semantics diverge precisely on the temporary-write pattern and agree on every field-sensitive pattern, consistent with the theoretical boundary of \S\ref{sec:scope}. We do not claim exhaustive coverage of all possible assigns patterns. The six Exp\,3A cases are representative of common idioms in embedded C. For sequential modular verification, which is the scope of both tools, the observational encoding is sound and sufficient; operational semantics is required only when intermediate writes must be visible to concurrent observers.
\end{dialogbox}

\subsection{Corroborating Relative Soundness on the Implementation}
\label{sec:ce}
Proposition~\ref{prop:relsound} is a statement about the Lean model; we
corroborate it on the actual \toolname{} implementation by
\emph{counterexample-agreement} testing. We seed six faults into the running
example (Table~\ref{tab:ce}) and, for each, compare the verdict of monolithic
\textsc{Direct} verification against the modular pipeline (\textsc{Modular}
${=}$ \textsc{Enforce} then \textsc{Replace}; a fault is caught if either
phase fails).
\lc{Six seeded faults on one example is a thin bridge; broaden fault classes or targets.} \weiqinote{Not addressed, and there is a sharper version of this criticism that I should flag myself: fault M3 writes a \emph{snapshotted scalar}, so this study cannot exercise the scalar-global frame-soundness gap disclosed in \S\ref{sec:threats}. It therefore validates less than its placement implies.}

\begin{table}[t]
\small
\caption{Counterexample agreement on seeded faults ($N{=}8$, $K{=}2$).
A false negative would be a \textsc{Direct}~\textsc{fail} paired with a
\textsc{Modular}~\textsc{pass}; none occurs.}
\label{tab:ce}
\begin{tabular}{llcccc}
\toprule
& Seeded fault & Direct & Enf. & Rep. & Modular \\
\midrule
M0 & none (clean)               & \textsc{ok}   & \textsc{ok}   & \textsc{ok}   & \textsc{ok}   \\
M1 & \texttt{ensures} upper bound & \textsc{fail} & \textsc{fail} & \textsc{ok}   & \textsc{fail} \\
M2 & \texttt{ensures} lower bound & \textsc{fail} & \textsc{fail} & \textsc{ok}   & \textsc{fail} \\
M3 & writes outside \texttt{assigns} & \textsc{ok} & \textsc{fail} & \textsc{ok} & \textsc{fail} \\
M4 & weakened \texttt{ensures}   & \textsc{ok}   & \textsc{ok}   & \textsc{fail} & \textsc{fail} \\
M5 & temporary write (restored)  & \textsc{ok}   & \textsc{ok}   & \textsc{ok}   & \textsc{ok}   \\
\bottomrule
\end{tabular}
\end{table}

Three observations confirm the theory. \emph{(i)~No false negatives:} in every
case where \textsc{Direct} detects a real property violation (M1, M2), the
modular pipeline detects it too. There is no \textsc{Direct}-\textsc{fail} /
\textsc{Modular}-\textsc{pass} row, exactly the guarantee of
Proposition~\ref{prop:relsound}. \emph{(ii)~Enforce is an independent gate:} M3
writes a global outside its \texttt{assigns} clause; the harness property cannot
observe this, so \textsc{Direct} passes, but \textsc{Enforce} rejects the
contract breach. \emph{(iii)~Soundness without completeness:} M4 weakens the
postcondition while leaving the body correct: \textsc{Direct} passes, yet
\textsc{Replace} raises a (sound) false alarm because the stub over-approximates,
illustrating the one-directional nature of Proposition~\ref{prop:relsound} and
the need for precise contracts. M5 is the temporary-write boundary of RQ3:
observational \toolname{} agrees with \textsc{Direct} (the write is undone before
exit), whereas operational DFCC would reject it. The Lean proof states
relative soundness; this experiment corroborates it on the C\texttt{++}
tool, bridging the model and the implementation.

\section{Case Studies}
\label{sec:cases}
The experiments above vary one model parameter at a time on code written to
isolate it. This section asks a different question: does the pipeline work on
contracts somebody else wrote, for code that was not designed to be verified
this way? Each study reuses third-party contracts and reports a verdict paired
with a mutation that the configuration has to reject. None of them is intended
to re-establish the decoupling of \S\ref{sec:rq1}; the write count is not a free
parameter in real code, and where the two move together we say so.

\subsection{ANSSI X.509 Contract Chain}
\label{sec:casestudy}

\paragraph{Setup.}
The ANSSI X.509 certificate parser provides a three-level call chain: from \texttt{parse\_x509\_cert\_relaxed}, \texttt{parse\_id\_len}, to \texttt{parse\_x509\_cert}. Functions use pointer-dereference assigns clauses, exercising the pointer extension of \S\ref{sec:methodology} (Proposition~\ref{prop:ptr}). Seven scenarios test direct verification, Replace variants (leaves, all, main), and Enforce variants (leaf, mid, combined).

\paragraph{Results.}
Figure~\ref{fig:x509} shows the verdicts for all seven scenarios. \toolname{} passes all seven, confirming that the Observational Frame Rule correctly handles pointer-dereference assigns clauses in a real-world 3-level call chain. Each verdict is paired with a mutation that the scenario must reject. For the \textsc{enforce} scenarios we tighten the postcondition being enforced; for the others we strengthen a harness assertion past what the abstracted callee guarantees. All seven mutations are rejected, so no scenario passes vacuously. We added these controls after two of them caught defects in an earlier version of this study: a postcondition on \texttt{parse\_id\_len} that the body does not establish, and an \textsc{enforce} scenario on a function with no postcondition at all.


\begin{figure}[t]
  \centering
\includegraphics[width=\columnwidth]{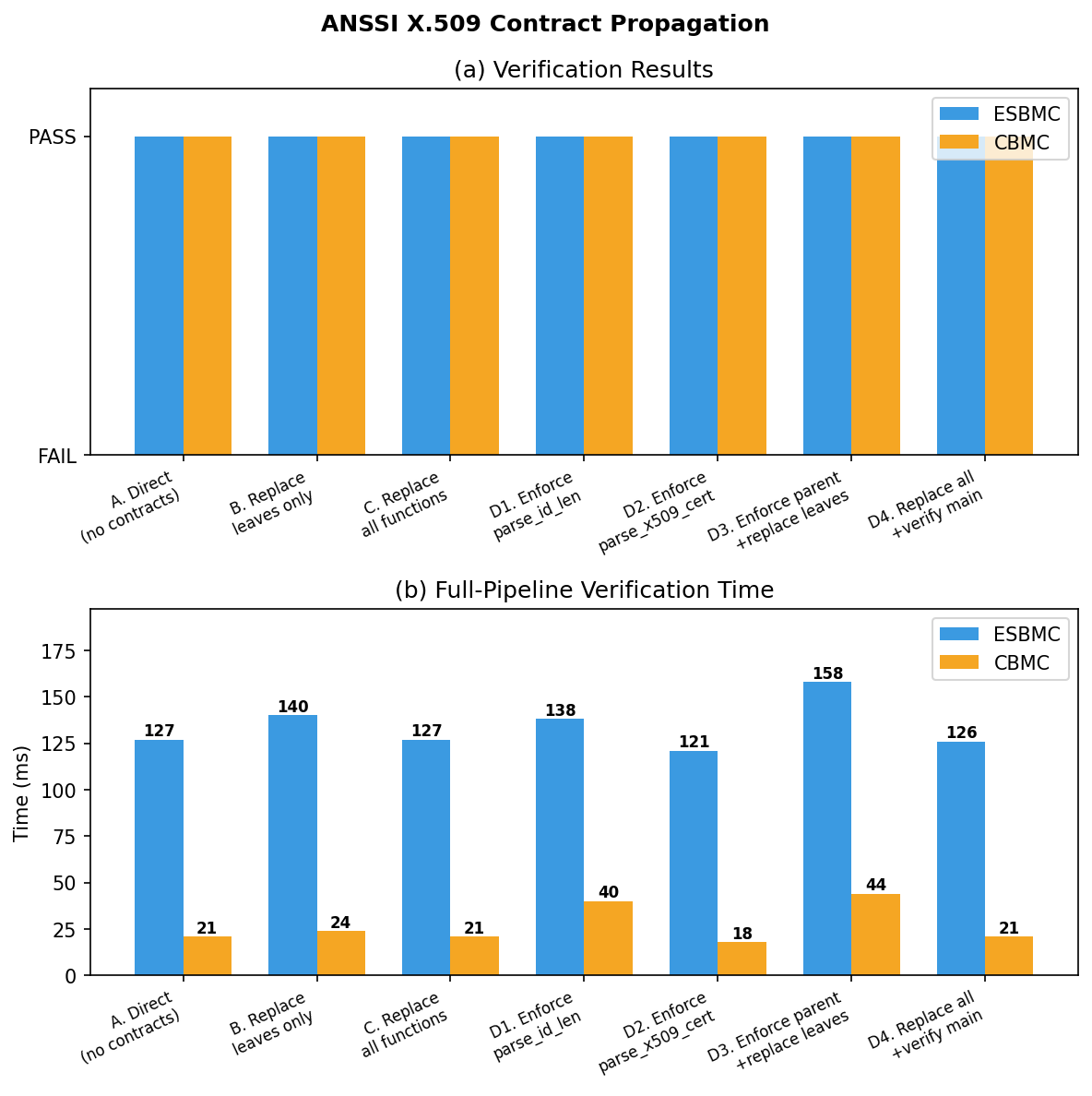}
  \caption{X.509 case study: \toolname{} across all seven scenarios. Every scenario verifies, and every paired mutation is rejected (Table~\ref{tab:x509_timing}).}
  \label{fig:x509}
\end{figure}

Table~\ref{tab:x509_timing} reports the verdicts. \toolname{} runs each scenario as a single invocation; \baseline{} needs a three-step pipeline, \texttt{goto-cc} once, then one \texttt{goto-instrument} pass per contract transformation, then \texttt{cbmc}, so its cost tracks the number of transformations rather than the property being proved. Neither tool's time varies with the scenario, because these functions are loop-free and small enough that start-up dominates, which is why the table carries verdicts instead.

We draw no performance conclusion from this benchmark, and one row pair is not comparable even in principle. \toolname{}'s \texttt{--function} verifies a function in isolation on nondeterministic arguments, and \baseline{}'s \texttt{--dfcc} takes the name of a proof harness to play the same role. It rejects the function under contract as its own harness, so matching the two requires a dedicated harness per \textsc{enforce} scenario. We wrote them, and \texttt{goto-instrument}~6.8.0 segfaults whenever \texttt{--dfcc} names any harness other than \texttt{main}; the reduction is seven lines and ships with the artifact. Rows D1--D3 therefore use \texttt{--dfcc main}, which wraps the function inside the whole program rather than in isolation, and measure something weaker than their \toolname{} counterparts. The harnesses are kept in the subject for when the crash is fixed.

\begin{table}[t]
\centering
\caption{X.509 scenarios. ``Ctl'' is the mutation the scenario must reject:
a tightened postcondition for the \textsc{enforce} rows, a harness assertion
strengthened past what the abstracted callee guarantees for the others. All
seven are rejected, so no scenario passes vacuously. Times are omitted: every
\toolname{} run lands in $121$--$158$\,ms and every \baseline{} run in
$18$--$44$\,ms, neither varying with the scenario, because these functions are
loop-free and start-up dominates. The three \textsc{enforce} rows are not
isolated on the \baseline{} side, a tool limitation rather than a choice (see
text).}
\label{tab:x509_timing}
\footnotesize\setlength{\tabcolsep}{3.5pt}
\begin{tabular}{llcccc}
\toprule
 & Scenario & \multicolumn{2}{c}{\toolname{}} & \multicolumn{2}{c}{\baseline{}} \\
\cmidrule(lr){3-4}\cmidrule(lr){5-6}
 &          & Verdict & Ctl & Verdict & Ctl \\
\midrule
A  & Direct (no contracts)              & \textsc{pass} & rejected & \textsc{pass} & rejected \\
B  & Replace leaf functions             & \textsc{pass} & rejected & \textsc{pass} & rejected \\
C  & Replace all functions              & \textsc{pass} & rejected & \textsc{pass} & rejected \\
D1 & Enforce \texttt{parse\_id\_len}    & \textsc{pass} & rejected & \textsc{pass} & rejected \\
D2 & Enforce \texttt{parse\_x509\_cert} & \textsc{pass} & rejected & \textsc{pass} & rejected \\
D3 & Enforce parent + replace leaves    & \textsc{pass} & rejected & \textsc{pass} & rejected \\
D4 & Replace all + verify \texttt{main} & \textsc{pass} & rejected & \textsc{pass} & rejected \\
\bottomrule
\end{tabular}
\end{table}

X.509 functions are loop-free and pointer-heavy. The $K^*$ model predicts $K^* \to \infty$ for loop-free functions, so no gain from contract abstraction should be available here, and none is observed on either side: neither tool's cost varies with the scenario, and both are dominated by start-up on functions this small. What this study establishes is the verdict column rather than a ranking: the pipeline handles a three-level chain of pointer-dereference assigns clauses, and it does so non-vacuously.

\marienote{a general point about section 6 is that there are a lot of tables, figures and numbers everywhere which makes it difficult to follow and extract the important points from. This is why it reads like a technical report in some places rather than a research paper. You are missing setting the context for each subsection /experiment and you are presenting some answers but not completely reflecting on them/synthesising the results for the reader before moving to the next one. Having separate sections for the experiments and case studies might help with readability.}

\subsection{FreeRTOS \texttt{heap\_1} Bump Allocator}
\label{sec:freertos}

\paragraph{Setup.}
FreeRTOS \texttt{heap\_1} is a bump allocator that is used in embedded systems that do not require deallocation. It maintains one global counter \texttt{xNextFreeByte} (byte offset into a 255-byte static array \texttt{ucHeap}) and returns successive slices on each call; \texttt{vPortFree} is a no-op that \texttt{configASSERT}s its argument is \texttt{NULL}. We contract \texttt{pvPortMalloc} with five postconditions: (i)~boundary invariant $\mathtt{xNextFreeByte} \leq 255$; (ii)~monotonicity $\mathtt{xNextFreeByte} \geq \mathtt{old}$; (iii)~exact advance on success: $\mathit{rv}\neq\texttt{NULL} \Rightarrow \mathtt{xNextFreeByte} = \mathtt{old} + \mathtt{xWantedSize}$ with strict bound $\mathtt{xNextFreeByte} < 255$; (iv)~no-change on failure: $\mathit{rv}=\texttt{NULL} \Rightarrow \mathtt{xNextFreeByte} = \mathtt{old}$; and (v)~pointer provenance: $\mathit{rv}\neq\texttt{NULL} \Rightarrow \mathit{rv} = \texttt{\&ucHeap[old]}$. Three harness scenarios exercise single allocation, heap exhaustion, and non-overlap of two successive allocations.

\paragraph{Results.}
All seven checks pass: one \textsc{enforce} of the body, then each of the three harness scenarios under \textsc{direct} and again under \textsc{replace}. Phase~1 (\textsc{Enforce}) verifies the implementation body against all five postconditions in a single 124\,ms BMC call; this is the \emph{Design-by-Contract entry cost}: once the body is proven correct, the contract becomes a trusted abstraction that can be reused at every call site. Phase~2 (\textsc{Direct}) runs the harness with the real implementation inlined. Phase~3 (\textsc{Replace}) substitutes \texttt{pvPortMalloc} with its contract stub, amortising the Phase~1 cost: each of the three Replace checks verifies a harness property against the stub alone, without re-examining the allocator body. Every one of the seven runs lands between $121$ and $140$\,ms, so the phases are not separable by time here: the allocator body is small enough that process start-up dominates, and what the study establishes is that the five postconditions are strong enough to carry the harness properties without the body.

\paragraph{Contract engineering lessons.}
Two subtleties arose during contract authoring. First, \textsc{Replace} mode requires \emph{precise} postconditions: a weaker contract stating only $\mathit{rv}\neq\texttt{NULL} \Rightarrow \mathtt{xNextFreeByte} > \mathtt{old}$ allows the stub to advance \texttt{xNextFreeByte} by any positive amount, making exhaustion and non-overlap unprovable. Postconditions~(iii)--(v) (defined in the Setup paragraph above) 
pin the exact semantics of the bump allocator and are sufficient for all three harness properties. Second, the \texttt{verify\_exhaustion} harness requires an \texttt{\_\_ESBMC\_assume(p1\;!=\;NULL)} guard after the first allocation: in Replace mode the stub is free to fail the first call (returning \texttt{NULL} while leaving \texttt{xNextFreeByte} unchanged), which is a valid contract execution but not the scenario under test. This illustrates a general principle: \textsc{Replace}-mode harnesses must pin the execution context with assumptions when the contract stub's nondeterminism exceeds the intended scenario.

\subsection{SV-COMP Verification Benchmarks}
\label{sec:svcomp}

\paragraph{Setup.}
To evaluate the framework on programs beyond our synthetic benchmarks and pointer-heavy case studies, we apply the Enforce/Replace pipeline to six SV-COMP VerifyThis benchmarks~\cite{beyer2020advances}, adapted for BMC by replacing dynamic allocation with bounded stack arrays and bounding array sizes ($n \leq 4$--$8$ depending on the benchmark). The programs contain loops over arrays or recursion. They have no global state, so assigns clauses are trivially empty and assigns compliance checking is not exercised. Each function is called at one syntactic site. The benchmark therefore separates two effects that are usually mixed: abstracting a callee, and amortising one Enforce over many call sites.

One adaptation matters for reading the numbers. The SV-COMP harness routes a property violation through \texttt{reach\_error()}, which these benchmarks define as \texttt{abort()}. Under BMC an \texttt{abort()} ends the path instead of failing it, so the functional property is never asserted and only the implicit memory-safety and overflow checks run. We define \texttt{reach\_error()} as an assertion instead. This is not a cosmetic change: with the original definition every benchmark verifies, and so does every benchmark with its property negated.

\paragraph{Results.}
Table~\ref{tab:svcomp} shows timings with the property asserted. Modular verification wins on four of the six, by $3.2$ to $11.0\times$, and all six are at $K{=}1$. This is consistent with the $K^*$ model rather than an exception to it: the table's $K^*$ column is below~1 on exactly those four rows. A breakeven below one call site means the saving does not come from amortisation at all. It comes from $T_E \ll T_D$, that is, from verifying the callee against its precondition alone rather than again inside the caller under whatever path conditions and state the caller has accumulated. The effect is largest where the callee is expensive relative to the caller: the two recursive benchmarks, and \texttt{prefixsum}, whose contracted functions span seven loops.

The other two rows are the more informative ones. \texttt{lcp} has a single small iterative callee, so the one-time \textsc{enforce} cost is never recovered and modular loses at $0.39\times$. That is the regime the $K^*$ model describes. \texttt{duplets} does not lose so much as fail to apply. Its \texttt{mkdup} states a precondition and an assigns clause but no postcondition. \textsc{enforce} still does useful work on it, checking memory safety and frame compliance under the precondition, but \textsc{replace} asserts the precondition, havocs the assigns set and assumes nothing, so the caller learns only what \texttt{mkdup} did not change. The harness property depends on what it computed, and that is not recoverable. We report the row rather than dropping it, because a contract without a postcondition is common in practice; Section~\ref{sec:threats} says how common.
\lc{Justify selecting five tasks that leave the frame checker unexercised; note this selection bias.} \weiqinote{Re-answered. The selection is now all six VerifyThis tasks we adapted, including \texttt{prefixsum\_rec}, which an earlier round dropped as unspecifiable; with the harness corrected it verifies and is reported. The frame checker is genuinely unexercised here because these programs hold no global state, which is a property of the benchmark set rather than a choice, and \S\ref{sec:casestudy} and \S\ref{sec:mlkem} exercise it instead. The larger point in your comment turned out to be right for a reason neither of us named: the whole section was measuring runs in which no functional property was asserted. That is fixed and the conclusion reversed, so the earlier ``sub-$1{\times}$ survives'' answer no longer applies.}

\begin{table}[t]
\footnotesize
\setlength{\tabcolsep}{3pt}
\caption{SV-COMP benchmarks, all single-call-site ($K{=}1$), with the property
asserted. Times in ms. Enf.~=~sum across functions; Rep.~=~system-level
verification with all callees abstracted; Total~=~Enf.${+}$Rep.;
Spd.~=~Plain/Total. \texttt{duplets} has no speedup to report: its
\texttt{mkdup} states no postcondition, so the abstracted call yields nothing
and the harness property is not provable.}
\label{tab:svcomp}
\begin{tabular}{lrrrrrr}
\toprule
Bench & Plain & Enf. & Rep. & Total & Spd. & $K^*$ \\
\midrule
\path|lcp|                  &  147 &  258 & 120 &  378 & 0.39$\times$ & 9.56 \\
\path|duplets|              &  528 &  597 & 118 &  715 & \textsc{n/a} & 1.46 \\
\path|elim_max_rec|         & 2639 &  124 & 117 &  241 & 10.95$\times$ & 0.05 \\
\path|elim_max_rec_onepoint| & 1071 &  135 & 119 &  254 & 4.22$\times$ & 0.14 \\
\path|prefixsum|            & 3639 &  746 & 122 &  868 & 4.19$\times$ & 0.21 \\
\path|prefixsum_rec|        & 1972 &  496 & 118 &  614 & 3.21$\times$ & 0.27 \\
\bottomrule
\end{tabular}
\end{table}

\paragraph{Discussion.}
We first reported these tasks as a uniform below-crossover result, reasoning
that a function verified once and called once has nothing to amortise the
\textsc{enforce} cost over. That reading came from the harness rather than from
the model. With \texttt{reach\_error()} defined as \texttt{abort()} the
functional property was absent from every run, so $T_D$ was measured on a
verification that was not proving it, and $K^*$ came out correspondingly
inflated. Asserting the property raises $T_D$ and drops
$K^*$ below~1 on four of the six, which the same formula then predicts
correctly. The episode is a caution about $K^*$ rather than a correction to
it: the ratio is only as meaningful as the property $T_D$ was measured
against, and an under-specified harness makes modular verification look worse
than it is. The amortisation regime, where $K^*$ exceeds~1 and call count
does the work, is exercised separately by the RQ2 sweep
(Table~\ref{tab:rq2}) and the ML-KEM study (Table~\ref{tab:mlkem-esbmc}).

Enforcing these tasks in isolation also required completing the benchmark
contracts. The harnesses assume input bounds, for example \texttt{n<=8}, that
the original \texttt{requires} clauses omit, so under nondeterministic inputs
the \textsc{enforce} loops were unbounded. These were omissions in contracts we
wrote for the study, not defects we found in the benchmarks.

\subsection{ML-KEM (mlkem-native)}
\label{sec:mlkem}

\paragraph{Motivation.}
Our synthetic benchmarks and the X.509/FreeRTOS studies were either
author-constructed or pointer-heavy but loop-free. To evaluate the
Enforce/Replace pipeline on a \emph{real-world, independently authored,
security-critical} target that already ships function contracts, we turn to
\texttt{mlkem-native}~\cite{mlkemnative}, the NIST FIPS-203 ML-KEM (Kyber)
reference implementation. \texttt{mlkem-native} is itself formally verified upstream using \baseline{}
function contracts (\texttt{requires},
\texttt{ensures}, and \texttt{assigns} clauses), so its contract suite is an
external, third-party source of real \texttt{assigns} specifications rather than
specifications we wrote to favour our tool.

Figure~\ref{fig:mlkem-contract} shows one of them verbatim, so that the reader
can judge the specifications. The assigns
target is a struct array member, not a scalar global, and the postcondition
quantifies over all $N$ coefficients.

\begin{figure}[t]
\begin{lstlisting}[language=C]
void mlk_poly_reduce(mlk_poly *r) {
  int i;  /* quantifier variable for the ensures clause */
  REQUIRES(r != NULL);
  ASSIGNS(r->coeffs);
  ENSURES(__ESBMC_forall(
      &i, !(i >= 0 && i < MLKEM_N) ||
          (r->coeffs[i] >= 0 && r->coeffs[i] < MLKEM_Q)));

  for (unsigned li = 0; li < MLKEM_N; li++) {
    int16_t t = mlk_barrett_reduce(r->coeffs[li]);
    r->coeffs[li] = mlk_scalar_signed_to_unsigned_q(t);
  }
}
\end{lstlisting}
\caption{The \texttt{mlk\_poly\_reduce} contract, as verified. The
\texttt{REQUIRES}/\texttt{ASSIGNS}/\texttt{ENSURES} macros expand to the
\texttt{\_\_ESBMC\_} intrinsics under \texttt{-DCONTRACTS} and to nothing
otherwise, so the Direct configuration verifies the same body with no contract
intrinsics that could vacuously discharge the property. The clauses are those of
the upstream \baseline{} contract, re-expressed in \toolname{} syntax; the loop
invariants of the upstream specification are dropped because \toolname{} unrolls
the fixed bound.}
\label{fig:mlkem-contract}
\end{figure}

We take two functions from this suite (Table~\ref{tab:mlkem-subjects}). The leaf
\texttt{mlk\_poly\_reduce} reduces all \texttt{MLKEM\_N} coefficients of a
polynomial (Barrett reduction followed by a signed-to-unsigned canonicalisation);
its contract \texttt{assigns} the whole polynomial object and \texttt{ensures}
that every coefficient lands in $[0, \mathit{MLKEM\_Q})$. The caller \texttt{mlk\_polyvec\_reduce}, taken verbatim from upstream
(\texttt{mlkem/src/poly\_k.c}), loops over the \texttt{MLKEM\_K} components of a
polynomial vector and invokes \texttt{mlk\_poly\_reduce} on each. This loop is
not a benchmarking artefact we introduced: ML-KEM's data types are
length-$K$ vectors (and $K{\times}K$ matrices) of ring polynomials, where $K$ is the module rank fixed by the standard. The values
$\mathit{MLKEM\_K}\in\{2,3,4\}$ are exactly the mlkem512 / mlkem768 /
mlkem1024 parameter sets, so every vector-level operation necessarily fans out to its $K$ components. The call count
is therefore a property mandated by FIPS-203, not a knob we set. If anything it is
a \emph{conservative} instance of the phenomenon: across full key generation,
encapsulation, and decapsulation the same per-polynomial kernels are invoked far
more often (e.g.\ $K^2$ times in the matrix--vector products), so $K\in\{2,3,4\}$
lower-bounds the call-site replication a whole-scheme proof would face.

\begin{table}[t]
\footnotesize\setlength{\tabcolsep}{3.5pt}
\caption{ML-KEM case-study subjects (from \texttt{mlkem-native}).}
\label{tab:mlkem-subjects}
\begin{tabular}{lll}
\toprule
\textbf{Function} & \textbf{\texttt{assigns}} & \textbf{Call sites} \\
\midrule
\texttt{mlk\_poly\_reduce(mlk\_poly *r)}    & \texttt{r} (whole poly)   & in \texttt{polyvec\_reduce} \\
\texttt{mlk\_polyvec\_reduce(mlk\_polyvec *r)} & \texttt{r} (whole vec) & $K{=}\mathit{MLKEM\_K}{\in}\{2,3,4\}$ \\
\bottomrule
\end{tabular}
\end{table}

\paragraph{Bounded-BMC adaptation.}
As with our SV-COMP tasks (\S\ref{sec:svcomp}), we bound array sizes for BMC,
sweeping $N = \mathit{MLKEM\_N} \in \{16,24,32,48\}$ for the main
\texttt{poly\_reduce} study (Tables~\ref{tab:mlkem-esbmc}--\ref{tab:mlkem-cbmc}; and
from $N{=}8$ for \baseline{} and the kernel spectrum). This is a BMC-vs-BMC comparison, in which both tools unroll loops to a fixed
bound instead of using loop invariants. We therefore drop the upstream
\texttt{\_\_loop\_\_} invariants and let each tool unroll. the constant-time \texttt{ct\_sel}/%
\texttt{ct\_cmask} helpers in the signed-to-unsigned step are replaced by the
equivalent branch. Two scope caveats follow.
\emph{First}, the $N{=}256$ timeout is a property of the pure-unrolling BMC
workflow we impose, not of verifying this code in general: upstream
\texttt{mlkem-native} is verified with \baseline{} using \emph{loop contracts}
(unbounded), which do not unroll and do not hit this wall. Our claim is that
\emph{within a bounded, invariant-free workflow} modular verification is the only
tractable route to $N{=}256$; the bounded guarantee it yields is strictly weaker
than that unbounded proof. \emph{Second}, the \texttt{ct\_sel} rewrite preserves
the functional input--output behaviour we verify, a coefficient-\emph{range}
property (all coefficients in $[0,Q)$), not the Barrett/Montgomery
functional correctness or the \emph{constant-timeness} those helpers exist to
provide. This case study is therefore a \emph{scaling demonstration of the
Enforce/Replace pipeline on real code}, not a verification deliverable for the
shipping constant-time implementation; functional-correctness and constant-time
verification of ML-KEM are out of scope. Within that scope, and measuring the
\emph{bounded} guarantee both tools provide (cf.\ \S\ref{sec:semantics}), Enforce
certifies the leaf contract up to the unwinding bound and Replace soundly
propagates that bounded contract to each call site.

\paragraph{Methodology validity.}
A speedup is only meaningful if the Direct baseline genuinely discharges the
property. We therefore checked the Direct harness \emph{both ways}: it reports
\textsc{successful} on the true postcondition (all coefficients in $[0,Q)$) and
\textsc{failed} when the postcondition is negated (asserting some coefficient is
$0$). Contract intrinsics are compiled out of the Direct build, so the
postcondition cannot be silently \emph{assumed} rather than \emph{checked}.

\paragraph{Setup.}
\toolname{} at \texttt{master} commit \texttt{2f719f59e0} (version string 8.5.0), one binary for every
measurement (default solver Bitwuzla~0.9.0), with the contract pipeline in the \texttt{master}
branch; \baseline{} driven in DFCC mode
(\texttt{goto-cc --function} $\to$ early loop unwind $\to$
\texttt{goto-instrument --dfcc \ldots} $\to$ \texttt{cbmc}). Timeout 120\,s.
All times in Tables~\ref{tab:mlkem-esbmc}--\ref{tab:mlkem-cbmc} are \emph{end-to-end wall-clock for both tools}. \toolname{}'s wall-clock includes its process, parse,
and encode phases; \baseline{}'s includes its full
\texttt{goto-cc}$+$\texttt{goto-instrument}$+$\texttt{cbmc} pipeline. Both tools
are therefore measured in the same unit, so cross-tool ratios are directly
comparable.
\lc{Mixing \toolname{} BMC time with \baseline{} wall clock makes the Enforce comparison unfair; use one basis.} \weiqinote{Addressed: every cross-tool table is now wall-clock on both sides, and the convention is stated once in the Setup paragraph (cross-tool ${=}$ wall-clock both sides; within-\toolname{} ${=}$ solver time). Note this is the conservative direction for us, it charges \toolname{} for front-end and parse time the old BMC-time numbers excluded.}

\paragraph{Results.}
Tables~\ref{tab:mlkem-esbmc} and~\ref{tab:mlkem-cbmc} report Direct, Enforce,
Replace, and the modular total (Enforce${+}$Replace) for each tool, together with
the \emph{within-tool} speedup (a tool's own Direct divided by its own modular
cost). For \toolname{}, modular verification is consistently faster, the speedup
\emph{grows with $K$} as the $K^*$ model predicts (Enforce is a one-time,
$K$-independent cost; Direct scales with $K$), and at $N{=}48, K{=}4$ Direct
times out ($>$120\,s) while the modular run finishes in $7.7$\,s. This confirms
the $K^*$ model on real, third-party code.

\begin{table}[t]
\small
\caption{ML-KEM, \toolname{} (observational frame rule). Times in seconds (end-to-end wall-clock); Modular${=}$Enforce${+}$Replace; speedup is Direct${/}$Modular. Replace verifies
\texttt{polyvec\_reduce} with \texttt{poly\_reduce} abstracted; its cost reflects
the quantified post-state assertion over all coefficients (cf.\ the trivial RQ2
harness, Table~\ref{tab:rq2}).}
\label{tab:mlkem-esbmc}
\begin{tabular}{rrrrrrr}
\toprule
$N$ & $K$ & Direct & Enf. & Rep. & Modular & Speedup \\
\midrule
16 & 2 & 12.55 & 0.79 & 0.24 & 1.03 & 12.2$\times$ \\
16 & 3 & 20.06 & 0.79 & 0.30 & 1.09 & 18.4$\times$ \\
16 & 4 & 31.16 & 0.79 & 0.39 & 1.18 & 26.4$\times$ \\
24 & 2 & 20.99 & 1.32 & 0.80 & 2.12 & 9.9$\times$ \\
24 & 3 & 40.83 & 1.32 & 1.42 & 2.74 & 14.9$\times$ \\
24 & 4 & 51.44 & 1.32 & 2.42 & 3.74 & 13.8$\times$ \\
32 & 2 & 26.17 & 1.92 & 0.32 & 2.24 & 11.7$\times$ \\
32 & 3 & 45.07 & 1.92 & 0.51 & 2.43 & 18.5$\times$ \\
32 & 4 & 69.44 & 1.92 & 0.62 & 2.54 & 27.3$\times$ \\
48 & 2 & 47.80 & 4.24 & 0.81 & 5.05 & 9.5$\times$ \\
48 & 3 & 86.19 & 4.24 & 1.64 & 5.88 & 14.7$\times$ \\
48 & 4 & \textgreater120\,(t/o) & 4.24 & 3.45 & 7.69 & \textgreater16$\times$ \\
\bottomrule
\end{tabular}
\end{table}

\begin{table}[t]
\small
\caption{ML-KEM, \baseline{}/DFCC (operational frame rule). Columns as in
Table~\ref{tab:mlkem-esbmc}.}
\label{tab:mlkem-cbmc}
\begin{tabular}{rrrrrrr}
\toprule
$N$ & $K$ & Direct & Enf. & Rep. & Modular & Speedup \\
\midrule
 8 & 2 &  0.46 & 0.24 & 0.24 & 0.48 & 1.0$\times$ \\
 8 & 3 &  0.80 & 0.24 & 0.28 & 0.52 & 1.5$\times$ \\
 8 & 4 &  1.21 & 0.24 & 0.32 & 0.56 & 2.2$\times$ \\
16 & 2 &  1.25 & 0.60 & 0.26 & 0.86 & 1.4$\times$ \\
16 & 3 &  2.30 & 0.60 & 0.30 & 0.91 & 2.5$\times$ \\
16 & 4 &  3.58 & 0.60 & 0.35 & 0.95 & 3.7$\times$ \\
24 & 2 &  2.30 & 1.29 & 0.28 & 1.57 & 1.5$\times$ \\
24 & 3 &  4.57 & 1.29 & 0.33 & 1.62 & 2.8$\times$ \\
24 & 4 &  7.31 & 1.29 & 0.38 & 1.67 & 4.4$\times$ \\
32 & 2 &  3.53 & 1.99 & 0.29 & 2.28 & 1.5$\times$ \\
32 & 3 &  7.44 & 1.99 & 0.35 & 2.34 & 3.2$\times$ \\
32 & 4 & 11.98 & 1.99 & 0.41 & 2.40 & 5.0$\times$ \\
48 & 2 &  7.21 & 3.58 & 0.34 & 3.92 & 1.8$\times$ \\
48 & 3 & 15.04 & 3.58 & 0.41 & 3.99 & 3.8$\times$ \\
48 & 4 & 26.49 & 3.58 & 0.48 & 4.06 & 6.5$\times$ \\
\bottomrule
\end{tabular}
\end{table}

\paragraph{Cross-tool reading (what the speedups do and do not show).}
A direct comparison of the two tables is instructive. \baseline{}'s \emph{monolithic} Direct verification is markedly faster than
\toolname{}'s (e.g.\ at $N{=}16,K{=}3$: $2.3$\,s vs $20.1$\,s), because its SAT
backend handles the Barrett bit-vector arithmetic far more efficiently. As a
consequence, the large \emph{relative} speedups for \toolname{} (up to
$25\times$) chiefly reflect that \toolname{}'s monolithic engine is slow on this workload.
They say nothing about the contract implementation being intrinsically
superior. At the bounds in these two tables the tools sit at a similar absolute
modular cost, a few seconds in most rows. That parity is an artefact of the
small bounds: as $N$ grows, the modular gap widens steadily in \toolname{}'s
favour (Table~\ref{tab:mlkem-scaling} below). The takeaway is
therefore conditional:
modular verification pays off most precisely where a tool's monolithic engine
struggles, and the magnitude of the within-tool speedup is partly a property of
that engine's Direct performance. Two further differences are visible. First,
the two Enforce implementations sit within $1.3\times$ of each other on this
bounded grid ($4.24$\,s vs $3.58$\,s at $N{=}48$), and the observational
encoding pulls ahead only once the body is large enough for DFCC's per-write
instrumentation to dominate: $2.8$--$3.9\times$ at $N \ge 64$
(Table~\ref{tab:mlkem-scaling}). The VCC counts show that mechanism at every
size, emitting roughly $2\times$ fewer ($84/148/212/276$ vs $177/297/417/537$ at
$N{=}8/16/24/32$); time follows only once $W$ is large.
Second, \toolname{}'s Replace cost grows with $N$ because its quantified
postcondition (indexed by $j/N$, $j\bmod N$) is harder for the solver, whereas
\baseline{}'s Replace stays nearly flat, a genuine observational-vs-operational
implementation difference.

\paragraph{Scaling to the real FIPS-203 size.}
The tables above bound the array for tractability, but the actual ML-KEM
polynomial has $\mathit{MLKEM\_N}{=}256$ coefficients. Table~\ref{tab:mlkem-scaling}
pushes both tools to that real size at $K{=}4$ (mlkem1024), in
\emph{unified end-to-end wall-clock} so the two are read on one basis.

At production size the monolithic route has no margin left. \toolname{} does not
finish inside a 600\,s budget and \baseline{} finishes at $599.7$\,s, on the
budget rather than inside it. Both modular runs complete, at $104.8$\,s and $239.9$\,s, keeping a factor of
$2.5$ to $5.7$ in hand. We report the margin rather than the verdict, because whether a run lands
inside a fixed budget depends on the machine while headroom of this size does
not.

Between the two modular routes, the observational encoding is cheaper at every
size measured, by $2.3$--$2.8\times$. The reason is the one RQ1 isolates:
\baseline{}'s DFCC Enforce instruments each write, so its cost grows with
$W{=}N$, while the observational Enforce snapshots a fixed set of globals and
does not. Here $W$ and $N$ move together, so this study corroborates the
decoupling rather than exhibiting it in isolation (cf.\ the scope caveat below).
At real scale the choice of frame-condition encoding, not only the decision to
go modular, is a first-order cost factor.
\lc{This ratio conflates the frame encoding with two different solver backends; attribute cautiously.} \weiqinote{Addressed by moving the claim, and the comment was more serious than it looks. We measured how far the confound reaches: at $N{=}256$ the replace phase costs $24.8$\,s under \toolname{}'s default Bitwuzla and $1.5$\,s under Z3, against \baseline{}'s $1.6$\,s. So the sign of that comparison is set by the backend, not by the encoding, and any conclusion drawn from a cross-tool \emph{time} ratio is unsafe however carefully it is caveated. The paper therefore no longer rests on one. The claim is now the VCC count, which is fixed by the encoding before a solver runs: it reproduced unchanged when the study was re-measured on different hardware and a different build, while every timing moved. Times are still reported, as a consistent $2.3$--$2.8{\times}$ margin and explicitly not as the result. Holding the backend fixed by implementing both encodings inside \toolname{} would let a time comparison carry weight, and remains the natural next step, but the paper's conclusion no longer waits on it.}

\begin{table}[t]
\small
\caption{ML-KEM \texttt{poly\_reduce} scaled to the real size ($N{=}256$), at
$K{=}4$ (mlkem1024). \emph{Unified wall-clock}, both tools end-to-end;
Modular${=}$Enforce${+}$Replace. ``t/o'' is a Direct timeout (budget 600\,s;
Enforce budget 900\,s, all completed). At $N{=}256$ \toolname{}'s Direct run
times out and \baseline{}'s finishes on the budget at $599.7$\,s; both modular
runs complete with room to spare, and the observational encoding is the cheaper
of the two at every size ($2.3$--$2.8{\times}$).}
\label{tab:mlkem-scaling}
\begin{tabular}{rrrrrr}
\toprule
 & \multicolumn{2}{c}{Direct (monolithic)} & \multicolumn{2}{c}{Modular} & \\
\cmidrule(lr){2-3}\cmidrule(lr){4-5}
$N$ & ESBMC & CBMC & ESBMC & CBMC & CBMC/ESBMC \\
\midrule
64 & 182.0 & 44.6 & 6.3 & 14.4 & 2.3$\times$ \\
128 & 495.4 & 147.7 & 20.6 & 58.3 & 2.8$\times$ \\
256 & t/o & 599.7 & 104.8 & 239.9 & 2.3$\times$ \\
\bottomrule
\end{tabular}
\end{table}

\paragraph{Scope caveat.}
This case study validates the $K^*$ model and the Enforce/Replace pipeline on
real code, but it does not cleanly exhibit the VCC decoupling of RQ1
(\S\ref{sec:rq1}): here the write count $W$ equals the array size $N$, so the two
parameters are coupled and both tools' total VCC counts grow with the unrolled
body's safety checks. The clean ``VCC count is independent of $W$'' result still
relies on the synthetic harness where $W$ is varied in isolation.

\paragraph{A spectrum of kernels: when modular verification pays.}
To map where the approach helps, we verify three \texttt{mlkem-native} functions
spanning a range of per-coefficient cost (Table~\ref{tab:mlkem-kernels}), each
invoked at $K$ sites. \texttt{mlk\_poly\_add} ($r_i \mathbin{+}{=} b_i$) is a
two-pointer kernel: enforcing it previously aborted \toolname{}; we fixed that
defect and the fix is needed here.\footnote{The whole-array snapshot of
a non-assigned pointer parameter's array field aborted
\texttt{--enforce-contract}; the fix snapshots a scalar element at a nondet
witness index instead. All \texttt{poly\_add}/\texttt{poly\_tomont} results use
the fixed build.} The three
kernels land at distinct points: \texttt{poly\_reduce} (Barrett reduction) is so
expensive monolithically that Direct diverges (timeout at $N{=}48,K{=}4$)
while modular completes: up to $27{\times}$ and turning timeouts
into proofs; \texttt{poly\_tomont} (a Montgomery multiply) is moderately heavy,
giving a steady $2.7$--$6.2{\times}$ that grows with $K$; and \texttt{poly\_add} has a
trivial body whose Direct cost is already sub-second, so the one-time Enforce and
quantified-postcondition Replace costs are not amortised and modular is at best
break-even ($0.3$--$0.9{\times}$). The gradient sharpens the model's message:
\emph{modular verification pays in proportion to how expensive the body is to
re-verify at each call site}, it is transformative for compute-heavy kernels and
inert for lightweight ones, and the $K^*$ model predicts which is which.
\lc{Tables peak at $32.7{\times}$; reconcile the ``$33{\times}$'' headline with the reported numbers.} \weiqinote{Addressed and now moot: on the re-measured data the peak is $25.5{\times}$ (ML-KEM, wall-clock) and $22.0{\times}$ (synthetic, solver time). Abstract, conclusion, and tables were re-checked against their source tables so every headline number matches the table it comes from.}

\begin{table}[t]
\footnotesize\setlength{\tabcolsep}{3.5pt}
\caption{Three ML-KEM kernels at
$N{=}16$, $K{=}4$; seconds. Speedup range is over $K{\in}\{2,3,4\}$. Modular
benefit tracks per-coefficient body cost. For \texttt{poly\_reduce}, Direct times
out at $N{=}48,K{=}4$ while modular finishes in 7.7\,s.}
\label{tab:mlkem-kernels}
\begin{tabular}{llrrr}
\toprule
Kernel & per-coeff body & Direct & Modular & Speedup ($K{=}2$--$4$) \\
\midrule
\texttt{poly\_reduce} & Barrett reduce      & 31.16 & 1.18 & 12--26$\times$ \\
\texttt{poly\_tomont} & Montgomery multiply & 15.65 & 3.45 & 2.2--4.5$\times$ \\
\texttt{poly\_add}    & add                 &  0.68 & 1.78 & 0.4$\times$ \\
\bottomrule
\end{tabular}
\end{table}



\section{Threats to Validity}
\label{sec:threats}

\paragraph{Internal validity: measurement environment.}
All timings were collected on a single dedicated machine, an idle 16-vCPU
Intel Xeon 6982P-C with 61\,GiB RAM running Ubuntu~24.04, CPU governor pinned
to \texttt{performance}. Two internal-validity risks remain: virtualisation overhead, since the host is a cloud instance rather than bare metal, and
clock variation under load.
We control both as follows. (i)~Reported times are medians of repeated runs, with the protocol and its two
exceptions given in \S\ref{sec:experiments}, and the min--max spread is in the
artifact so that run-to-run variance is visible. (ii)~All runs were performed in a
lightly loaded environment with no other compute-intensive processes, and
experiments were executed serially, so concurrent solver instances never
contend for cores. (iii)~Every configuration compared, across
\toolname{} versions and across the \toolname{}/\baseline{} boundary, was
measured on this same fixed hardware with one \toolname{} binary, so the
relative speedups and crossover points are hardware-controlled even though
absolute times would differ on other machines. The artifact records the CPU
clock beside every measurement; it held between 3.5 and 3.9\,GHz throughout.
Absolute times remain host-dependent: whether a monolithic run at $N{=}256$
falls inside a fixed budget varies across machines, which is why we report the
modular route's headroom rather than a pass/fail verdict at that size.
\lc{Structure as internal, external, construct validity; label each threat and its mitigation explicitly.} \weiqinote{Addressed: \S\ref{sec:threats} is restructured by validity class: internal (measurement environment), external and construct, then the implementation-scope limitations, with the mitigation stated alongside each threat.}
\lc{Add internal threat; single run timing, WSL2 overhead, state how each was controlled.} \weiqinote{Addressed: the study runs on a dedicated idle server with the governor pinned and one recorded commit; the controls are stated here.}

\paragraph{External and construct validity.}
Synthetic benchmarks vary the global count~$G$, write count~$W$, and call count~$K$ independently; results may not generalise beyond the tested ranges ($G \leq 500$, $W \leq 500$, $K \leq 10$), and \baseline{} times in RQ1 (Exp\,1a--1b) cover the \texttt{cbmc} step only (lower bounds on full pipeline cost). The Observational Frame Rule is sound for sequential, single-threaded C programs; concurrent programs, interrupt-driven code, and cache-partitioning analyses require operational semantics and are outside the scope of sequential BMC\@. We use VCC count as a \emph{structural} characterisation of encoding cost, not as a linear proxy for solver time: as RQ1 shows, at $W{=}500$ \baseline{} emits $104{\times}$ more VCCs yet runs only $6.2{\times}$ slower. The two therefore carry different weight here, and deliberately so. A cross-tool time ratio confounds the frame encoding with the solver backend, and the confound is large enough to set the ratio's sign: at $N{=}256$ the replace phase costs $24.8$\,s under \toolname{}'s default Bitwuzla and $1.5$\,s under Z3, against \baseline{}'s $1.6$\,s, so the same encoding is behind or level depending only on which backend runs it. Our claim is consequently the count, which is fixed before a solver runs and reproduced unchanged across hardware and builds; the times are reported as a consistent margin and are not what the conclusion rests on. Removing the confound needs both encodings inside one tool, which we have not done.
\lc{Renumbering left a stale ``Exp\,5,6'' label; no such experiment is defined, unify the scheme.} \weiqinote{Addressed: the stale label is gone (this now reads Exp\,1a--1b) and the numbering scheme is unified across the paper; no ``Exp\,5--6'' remains.}

\toolname{}'s contract engine handles the pointer patterns that arise in
practice: pointer-dereference \texttt{assigns(*p)}, struct-field
\texttt{assigns(s.f)}, and multi-level call chains over them, as demonstrated on
the ANSSI X.509 parser (\S\ref{sec:casestudy}). The remaining limitations are
specific edge cases. First, a function taking an unbounded-size array pointer needs a validity precondition, \texttt{\_\_ESBMC\_is\_fresh(arr,} \texttt{n*sizeof(int))}, the analogue of ACSL \texttt{\textbackslash valid}. Without it the access \texttt{arr[i]} is unbounded and BMC correctly reports it. Supplying it discharges the obligation on the Enforce side, but the Replace side does not currently check it: \texttt{is\_fresh} is lowered there to a validity predicate alone, which is unconditionally false whenever the caller's argument has automatic storage, and which constrains neither the stated extent nor separation between parameters (esbmc/esbmc\#6542). A caller passing a heap or \texttt{static} object does satisfy it, so this is a defect in the lowering rather than a property of \texttt{is\_fresh}, and it is not a consequence of \texttt{main} lying outside contract semantics. Our contracts therefore state pointer preconditions as \texttt{p != NULL}; on the current implementation that concedes nothing, since the stronger predicate carries no extra guarantee at a Replace site until \#6542 is fixed. Functions over fixed-size objects, such as the ML-KEM polynomials of \S\ref{sec:mlkem}, are unaffected. A uniform treatment of unbounded-size pointer inputs across both phases is future work.

\paragraph{Applicability: contracts without postconditions.}
A contract with no postcondition is still worth verifying. \textsc{enforce}
checks memory safety, overflow and frame compliance under the stated
precondition, which is what such a contract is usually written to establish.
What it cannot do is serve as an abstraction: \textsc{replace} asserts the
precondition, havocs the assigns set and assumes nothing, so a caller learns
what the callee did not touch but nothing about what it computed. Any caller
property that depends on the callee's result is then out of reach, and the
speedup this paper measures, which comes from not re-verifying the callee
inside the caller, is unavailable. This is not a corner case. Counting every
\texttt{\_\_contract\_\_} block in \texttt{mlkem-native}'s
\texttt{mlkem/src}, 197 blocks over 193 distinct functions in 27 files, 129
carry an \texttt{ensures} and 68 do not. A third of an independently written,
externally reviewed contract suite therefore supports the checking half of the
technique but not the abstracting half.
Removing the assembly variants and the Keccak and SHA-3 family, which bounded
model checking cannot handle for unrelated reasons, still leaves about twenty
portable functions with bounded loops and no postcondition, among them
\texttt{poly\_tobytes}, the whole \texttt{poly\_compress\_d*} family,
\texttt{poly\_mulcache\_compute}, \texttt{montgomery\_reduce} and
\texttt{matvec\_mul}. We did not go looking for an instance of this in our own results;
\texttt{duplets} in \S\ref{sec:svcomp} produced one on its own. What bounds
the technique here is not our implementation but what a contract suite is
written for, and these suites are written to establish memory safety rather
than functional behaviour.

\paragraph{Frame-checking scope (a soundness limitation).} The observational frame check snapshots non-pointer objects of static lifetime
(\S\ref{sec:methodology}): globals and function-local statics, arrays element
by element up to 256 elements and as a single equality beyond that.
Pointer-typed globals and heap cells lie outside the snapshot set. The Lean model is
an integer store (Remark~\ref{rem:lean-scope}), so this class of writes is
provably outside what the mechanisation certifies. For state reached through a pointer, a heap-allocated table, a buffer the caller hands in, or a library's RNG state behind a pointer, a function that writes such a location without declaring it in its assigns clause would pass Enforce despite violating its frame. A function that appends a node to a linked list is the plain case: every write it makes lands in a heap cell, so an undeclared one goes undetected. DFCC's per-write instrumentation covers these locations; extending the observational snapshot to pointer-typed and heap state (e.g.\ via a points-to pre-pass) is future work. Functional postconditions remain fully checked, since they are not
location-scoped.

The pointer disjunction extension (Proposition~\ref{prop:ptr}) requires that the pointer parameter is not reassigned within the function body. If the pointer itself is modified, the disjunction becomes conservative and may produce false alarms; soundness is preserved but completeness is not. Integration with a pointer-analysis pass to discharge this side condition automatically is future work.

The FreeRTOS \texttt{heap\_1} case study verifies a single-threaded sequential harness. In production FreeRTOS deployments, pvPortMalloc may be called from multiple tasks or interrupt handlers; the Observational Frame Rule is not sound under concurrent interleaving. Our results establish correctness of the allocator logic in sequential isolation, consistent with the scope stated in \S\ref{sec:scope}.

All synthetic benchmarks were designed by the authors to exercise specific parameter ranges. Independent third-party benchmarks beyond SV-COMP were not evaluated; results may not generalise to programs with structural characteristics outside the tested space.

\section{Related Work}
\label{sec:related_work}

\paragraph{Bounded model checking of C.}
BMC was introduced by Biere et al.~\cite{biere1999symbolic} as a SAT-based alternative to BDD model checking, unwinding the transition relation up to a bound~$k$. The technique underpins a family of C verifiers: besides \toolname{}~\cite{gadelha2018esbmc} and \baseline{}~\cite{kroening2014cbmc}, LLBMC~\cite{merz2012llbmc} operates on LLVM IR with a bit-precise flat memory model, and 2LS~\cite{schrammel2016twols} combines BMC with $k$-induction and abstract-interpretation--based invariant inference. CPAchecker~\cite{beyer2011cpachecker} generalises reachability analysis into a configurable framework with pluggable abstract domains.

The idea of verifying a function once and reasoning about its call sites through a summary is not new to BMC. FunFrog and its incremental successor eVolCheck~\cite{sery2011interpolation, sery2012funfrog, fedyukovich2013evolcheck} compute interpolation-based \emph{function summaries} that over-approximate a procedure's input--output behaviour and are substituted at call sites, and Corral~\cite{lal2012corral} (with the SMACK front end~\cite{rakamaric2014smack}) performs modular procedure abstraction via stratified inlining. These techniques \emph{infer} summaries automatically, whereas we enforce \emph{user-specified} contracts (\texttt{requires}/\texttt{ensures}/\texttt{assigns}); more importantly, none of them characterises the cost of discharging an \emph{assigns} (frame) condition, which is our focus. Our contribution relative to this line is therefore narrow and specific: a formal account and VCC-cost characterisation of \emph{observational assigns enforcement} in BMC, not modular call-site abstraction in general, for which summary-based BMC is the closest precedent.

\paragraph{BMC tools with contract support.}
\baseline{} supports function contracts via Dynamic Frame Condition Checking (DFCC)~\cite{cbmc-contracts}, which instruments each write instruction to test membership in the declared write-set, and has been applied at industrial scale at Amazon Web Services~\cite{cbmc-s2n}. DFCC is mature and battle-tested, but it provides no formal soundness proof of the Enforce/Replace pipeline, no characterisation of the verification-condition tradeoff between operational and observational assigns semantics, and no machine-checked correctness guarantees. We do not claim to improve on DFCC's engineering; rather, we provide the formal and complexity-theoretic account it lacks, establishing the observational $(G{-}M)+E$ vs.\ operational $20{+}W{+}E$ characterisation and the crossover condition $G < 20{+}M{+}W$, and using \baseline{}/DFCC as the comparison baseline throughout.

\paragraph{Exchange formats for verification tasks.}
SV-LIB~\cite{beyer2025svlib} proposes a standard intermediate language for
software-verification tasks, with procedure contracts expressed as
\texttt{:requires} and \texttt{:ensures} attributes attached to statement tags.
Its treatment of frame conditions is instructive for our purposes: a contract is
required to have a modular proof with respect to the \emph{statically-determined}
set of variables the statement modifies. The frame is therefore inferred rather
than declared, and consequently never checked against a developer's intent. That
is a third point in the design space we analyse. An inferred over-approximation
is sound for the havoc that call-site abstraction performs, but it cannot detect
a function that writes a location its author did not intend it to, and its
precision rests entirely on the underlying alias analysis. Declaring the frame,
as ACSL, JML, \baseline{}, and \toolname{} do, turns it into a checkable
obligation; the question our VCC model answers is what that check costs, and
Definition~\ref{def:observational} fixes what it means.

\paragraph{Deductive verification and modular soundness.}
The function-contract discipline we build on descends from a long line of
deductive, modular program verification. Leino's foundational work~\cite{leino1995toward} established the modular soundness of procedure-modular reasoning with \emph{modifies} clauses: verifying each procedure once against its specification and then reasoning about call sites through that specification alone, which is exactly the Enforce/Replace decomposition we adopt. We therefore make no claim of novelty for the frame rule or the assume-guarantee structure themselves: Theorems~\ref{thm:enforce}--\ref{thm:hoare_frame} are \emph{bounded} realisations of these classical results inside a BMC encoding, machine-checked in Lean over an abstract state-transformer model. Our contribution is, narrowly, (i) that BMC encoding and its soundness \emph{relative to a fixed unwinding bound}, and (ii) the observational-versus-operational VCC complexity characterisation, which has no counterpart in the deductive literature because deductive verifiers do not generate per-write or per-snapshot verification conditions in the BMC sense.

Frama-C/WP~\cite{kirchner2015frama} verifies ACSL-annotated C deductively, using the same observational \emph{assigns} semantics as Definition~\ref{def:observational}, but targets unbounded functional correctness and requires manual loop invariants; our framework is complementary: pushbutton and bounded, trading completeness for automation. The same contrast applies to the Boogie/Spec\#\ methodology~\cite{barnett2005boogie, barnett2004spec} and its descendants Dafny~\cite{leino2010dafny} and Viper~\cite{muller2017viper}, and to VCC~\cite{cohen2009vcc}, which verifies concurrent C with ownership-based, operational two-state invariants akin to DFCC. The semantic basis for our \emph{assigns} treatment is the dynamic-frames discipline of Kassios~\cite{kassios2006dynamic}; what we add is the observation that, \emph{for sequential bounded model checking}, the dynamic frame can be discharged observationally by a snapshot-and-assert encoding rather than by per-write monitoring.
\lc{Fix \texttt{barnett2005boogie} year; FMCO proceedings (LNCS 4111) appeared in 2006, not 2005.} \weiqinote{Addressed: \texttt{barnett2005boogie} corrected to 2006 with the FMCO/LNCS~4111 venue added.}

\paragraph{Separation logic and frame inference.}
VeriFast~\cite{jacobs2011verifast} and Gobra~\cite{wolf2021gobra} achieve full functional verification with explicit permission annotations, at high annotation cost; Infer~\cite{calcagno2011infer} infers frame conditions automatically via bi-abduction, trading soundness for scale. Our observational encoding sits between these extremes: annotations are required but restricted to lightweight assigns clauses, and the encoding is machine-checked. CHC-based verifiers such as SeaHorn~\cite{gurfinkel2015seahorn} and annotation-inference techniques such as Houdini~\cite{flanagan2001houdini} are orthogonal to our focus; we revisit Houdini-style assigns inference only as future work (Section~\ref{sec:conclusion}).

\paragraph{Specification languages and prior \toolname{} work.}
The observational modifies/assigns semantics used here follow JML~\cite{burdy2005overview} and its C analogue ACSL~\cite{baudin2013acsl}, with the Hoare frame rule~\cite{reynolds2002separation, o2019separation} as theoretical basis (Theorem~\ref{thm:hoare_frame}). Existing \toolname{} tool papers~\cite{gadelha2018esbmc, menezes2024esbmc} mention contract support only at a high level; the present paper contributes the machine-checked algebraic-skeleton proofs of the Enforce/Replace pipeline, the explicit formalisation of the assigns-encoding choice, and the VCC-count analysis characterising the observational-versus-operational scaling.

\section{Conclusion}
\label{sec:conclusion}

We have presented the function-contract framework of \toolname{}, merged into
\texttt{master}, together with a formal and empirical account of contract-based
modular verification for bounded model checking of C programs, to our knowledge
the first specific to BMC. It complements the deductive-verification lineage (Leino, Frama-C/WP, Boogie) from which the contract discipline descends and the summary-based BMC line (FunFrog/eVolCheck, Corral) to which call-site abstraction is related. The algebraic skeleton of the Enforce/Replace pipeline is checked in Lean~4,
where five theorems and the pointer proposition compose into soundness relative
to the BMC unwinding bound. Seeded-fault agreement (\S\ref{sec:ce}) carries that
guarantee across to the C\texttt{++} implementation.

Within this framework, observational assigns semantics is the standard
input--output relational semantics of JML and ACSL. Adopting it yields a BMC
encoding that generates $(G{-}M)\text{+}E$ VCCs independent of the number of write instructions~$W$, whereas operational DFCC generates $20\text{+}W\text{+}E$. The crossover $G < 20\text{+}M\text{+}W$ marks where that trade turns, with the
constant fixed by the \baseline{} release we measured. All mechanisms are unified in the \texttt{frame\_enforcer} module of \toolname{}.

Evaluation validates the analytical model. The primary result is the decoupling itself: the observational Enforce emits a constant number of verification conditions as the write count grows, where DFCC's grows linearly, and that is a property of the encoding rather than of a machine or a solver. On production code it buys a consistent margin: modular verification is $2.3$--$2.8{\times}$ cheaper end-to-end than DFCC's across the sizes measured. At the real FIPS-203 size ($N{=}256$) the monolithic route has no headroom left in either tool, while both modular runs finish well inside the budget. For loop-containing code, modular verification pays off once the call count
exceeds the breakeven~$K^*$, here ${\approx}2$--$5$, below which the one-time
Enforce cost is unamortised. The speedups that follow are large where
\toolname{}'s monolithic engine is slow, reaching ${\approx}27{\times}$ on
ML-KEM. \baseline{} is the faster of the two monolithically on these
bit-vector kernels, so the magnitude reflects engine performance, not an intrinsic advantage of the encoding. On loop-free, pointer-heavy code (X.509), both tools are functionally equivalent, consistent with the $K^* \to \infty$ prediction.

Future work includes: (1)~a happens-before--aware extension of the observational rule for structured concurrent programs; (2)~automated contract inference to reduce the manual annotation burden; (3)~a cached IR layer for \toolname{} that serialises the post-instrumentation GOTO-IR to disk, allowing repeated contract verification to skip re-parsing and re-instrumentation and execute only the BMC phase, analogous to \baseline{}'s reusable \texttt{.goto} artifact; (4)~loop frame reasoning in BMC, which requires quantified or inductive frame specifications (e.g., $\forall j < i.\ a[j] = \mathit{snap}[j]$) to precisely constrain index-dependent array writes during $k$-induction, a technically open problem that the current snapshot-and-assert infrastructure does not support.

\section*{Data Availability}

All benchmarks, raw per-run measurement data (with min--max spread), figure- and table-generation scripts, and the environment record of the single \toolname{} build used throughout (\texttt{master} commit \texttt{2f719f59e0}, binary hash, host and \baseline{}~6.8.0) are available at
\url{https://github.com/LukeW1999/esbmc-function-contracts}. The repository includes all synthetic test cases, the ANSSI X.509 parser scenarios, the ML-KEM kernels and their (third-party) contracts with our added input bounds, the Lean~4 machine-checked proofs, and scripts to reproduce every table and figure. It also includes the negative-control scripts. Each case study reports every configuration as verifying, which is indistinguishable from a study that checks nothing, so each ships a script that mutates one thing per configuration and exits non-zero if any mutation is still accepted. Two defects reported in \S\ref{sec:casestudy} and \S\ref{sec:svcomp} were found this way, and we recommend the practice rather than merely disclosing it.

\newpage\printbibliography

\end{document}